%% file: manuscript_arxiv.tex
\documentclass[journal,twoside,web]{ieeecolor}
\usepackage{noLogo}

\usepackage{cite}
\usepackage{amsmath,amssymb,amsfonts}
\usepackage{algorithmic}
\usepackage{graphicx}
\usepackage{textcomp}
\def\BibTeX{{\rm B\kern-.05em{\sc i\kern-.025em b}\kern-.08em
    T\kern-.1667em\lower.7ex\hbox{E}\kern-.125emX}}

\input{text_pre.tex}

\usepackage{hyperref}
\hypersetup{
    colorlinks=false,
    pdfborder={0 0 0},
    pdfborderstyle={/S/U/W 0},
}
\makeatletter
    \def\@endtheorem{
		 \vspace{0pt}\hfill  $\blacktriangleleft$ 
		\endtrivlist\@endpefalse }
\makeatother

\begin{document}
\title{What ODE-Approximation Schemes of Time-Delay Systems Reveal about Lyapunov-Krasovskii Functionals}
\author{Tessina H.\ Scholl, Veit Hagenmeyer, and Lutz Gr\"oll
\thanks{
\vspace{3em}%placeholder
%This paragraph of the first footnote will contain the date on 
%which you submitted your paper for review. It will also contain support 
%information, including sponsor and financial support acknowledgment. For 
%example, ``This work was supported in part by the U.S. Department of 
%Commerce under Grant BS123456.'' 
}
\thanks{The authors are with the Institute for Automation and Applied Informatics, Karlsruhe Institute of Technology, 76021 Karlsruhe, Germany (e-mail:\ tessina.scholl@kit.edu; veit.hagenmeyer@kit.edu; lutz.groell@kit.edu). }
}

\newpage\null
\noindent
\begin{tabular}{p{17cm}}
This article has been accepted  for publication in the IEEE Transactions on Automatic Control (tentatively scheduled to appear in the 2024 July issue),    
\href{https://doi.org/10.1109/TAC.2023.3347497}{\underline{doi:  10.1109/TAC.2023.3347497}}. 
\end{tabular}
\vfill
\begin{tabular}{|p{17cm}|}
\hline 
\copyright 2023 IEEE. Personal use of this material is permitted. Permission
from IEEE must be obtained for all other uses, in any current or future
media, including reprinting/republishing this material for advertising or
promotional purposes, creating new collective works, for resale or
redistribution to servers or lists, or reuse of any copyrighted
component of this work in other works.
\\
\hline
\end{tabular}
\thispagestyle{empty}\newpage
\setcounter{page}{0}

\maketitle

\begin{abstract}
\input{text_abstract}
\end{abstract}

\begin{IEEEkeywords}
delay systems, Lyapunov-Krasovskii functional, operator-valued Lyapunov equation, spectral methods, pseudospectral discretization, Gauss quadrature 
%Enter key words or phrases in alphabetical 
%order, separated by commas. For a list of suggested keywords, send a blank 
%e-mail to keywords@ieee.org or visit \underline
%{http://www.ieee.org/organizations/pubs/ani\_prod/keywrd98.txt}
\end{IEEEkeywords}

\input{text_main.tex}

\appendix

\input{text_appendix.tex}

\bibliographystyle{ieeetran}
\bibliography{Literatur}

\vfill
\end{document}

%% file: text_pre.tex
\usepackage{amsmath}
\usepackage{amssymb}
\usepackage{mathrsfs}   
\usepackage{color}
\usepackage{subfig}
\usepackage{multirow}
\usepackage{refcount} %\footnotemark[\getrefnumber{fn:labelname}] since \footnotemark[\ref{fn:labelname}] does not work in arxiv

\usepackage{etoolbox} %if...
\newtoggle{IsTwocolumn}
\settoggle{IsTwocolumn}{true}

\usepackage{tikz}
\usetikzlibrary{arrows.meta,bending}
\usetikzlibrary{calc}

\newtheorem{definition}{Definition}[section]
\newtheorem{theorem}{Theorem}[section]
\newtheorem{remark}{Remark}[section]
\newtheorem{example}{Example}[section]
\newtheorem{corollary}{Corollary}[section]
\newtheorem{lemma}{Lemma}[section]
\newtheorem{proposition}{Proposition}[section]
\newtheorem{condition}{Condition}[section]

\newcommand{\delay}{{h}}
\newcommand{\LM}{{\Psi}}
\newcommand{\Leg}{{p}}
\newcommand{\ind}{}

\newcommand{\N}[1]{{\text{\scalebox{0.75}{$[#1]$}}}}

\newcommand{\method}[1]{\textbf{\textit{\color{mblue}#1}}}

\allowdisplaybreaks  %page break in align

%% file: text_abstract.tex
The article proposes an approach to complete-type and related Lyapunov-Krasovskii functionals that neither requires knowledge of the delay-Lyapunov matrix function nor does it involve linear matrix inequalities. The approach is based on ordinary differential equations (ODEs) that approximate the time-delay system. The ODEs are derived via spectral methods, e.g., the Chebyshev  collocation method (also called pseudospectral discretization) or the  Legendre tau method. A core insight is that the Lyapunov-Krasovskii theorem resembles a theorem for Lyapunov-Rumyantsev partial stability in ODEs. For the linear approximating ODE, only a Lyapunov equation has to be solved to obtain a partial Lyapunov function. The latter approximates the Lyapunov-Krasovskii functional. Results are validated by applying Clenshaw-Curtis and Gauss quadrature to a semi-analytical result of the functional, yielding a comparable finite-dimensional approximation. In particular, the article provides a formula for a tight quadratic lower bound, which is important in applications. Examples confirm that this new bound is significantly less conservative than known results.  

%% file: text_main.tex
\section{Introduction}
Whenever a control law $u=\gamma(x)$ is constructed for a system $\dot x=f(x,u)$, 
the closed loop  description $\dot x(t) = f(x(t),\gamma(x(t))$ 
hinges on 
the availability of the instantaneous $x(t)$.
In practice, however, measurements, network communication, computation times, or the actuator response  cause a delay. The resulting $\dot x(t) = f(x(t),\gamma(x(t-h)))$, with a time delay $h>0$, is a retarded functional differential equation (RFDE) and can no longer  be tackled by the well-known stability theory of finite-dimensional ordinary differential equations (ODEs). What changes?

\subsection{Motivation: Delay-free versus Time-Delay System}\label{sec:motivation}
For delay-free nonlinear time-invariant ODEs  we could consider the linearization  about the equilibrium (provided it is hyperbolic and the right-hand side is differentiable) and simply conclude exponential stability from the eigenvalues of  $A$ in the resulting $\dot x =A x$, $A\in \mathbb R^{n\times n}$. We might be interested in the domain of attraction of the equilibrium. To this end, we could calculate a quadratic Lyapunov function $V(x)=x^\top P x$ for the linearized system
by prescribing a desired Lyapunov function derivative $D_{(\dot x =Ax)}^+V(x)=-x^\top Q x$. Solving the associated Lyapunov equation $PA+A^\top P=-Q$ for the matrix $P$ with a standard algorithm 
is accomplished 
in one line of Matlab code.
The obtained Lyapunov function also gives a negative Lyapunov function derivative in the nonlinear system -- at least in a certain domain around the equilibrium \cite{Khalil.2002}.
Let this domain be estimated by a norm ball 
with radius $r>0$. Then the probably most basic estimation of the domain of attraction, cf.\ \cite[Sec.~8.2]{Khalil.2002}, is provided by the set of points $x\in \mathbb R^n$ such that $V(x) <k_1 r^2$, 
where $k_1$ is the coefficient of the positive-definiteness bound 
$k_1\|x\|_2^2\leq V(x)$. Thus, having a non-conservative result for $k_1$ is important. 
It is simply the minimum eigenvalue of $P$   that provides the  largest possible coefficient $k_1$.

In time-delay systems, analogous steps become more elaborate. Given a nonlinear system, the principle of linearized stability still holds \cite{Diekmann.1995}, and we are led to the linear RFDE
\begin{align} \label{eq:lin_RFDE}
\dot x(t) &= A_0 x(t) + A_1 x(t-\delay), 
\end{align} 
$A_0,A_1\in \mathbb R^{n\times n}$, 
with a discrete delay $\delay> 0$. 
The characteristic equation has, generically, an infinite number of roots.  
Still, 
to determine stability for a given delay $\delay$, we can resort to numerical eigenvalue calculations \cite{Breda.2015,Wu.2012,Vyhlidal.2014b,Engelborghs.2002,Ito.1985}
or we use other characteristic-equation-based criteria that 
can prove  stability for all delays or all delays smaller than a first critical one \cite{Gu.2003,Scholl.2023}. 
The initial state $x_0$ is in fact an initial function, 
the domain of attraction is a set of initial functions, the state $x_t$ at time $t\geq 0$ represents the solution segment on the past delay interval $[t-\delay,t]$, and 
instead of a Lyapunov function,  a Lyapunov-Krasovskii (LK) functional $V(x_t)$ is required. 
Analogously to the delay-free case discussed above, we can explicitly prescribe the desired LK functional derivative and determine the corresponding LK functional. 
The
LK functional derivative along trajectories of (\ref{eq:lin_RFDE}) is  commonly  \cite{Kharitonov.2013} 
set as 
\iftoggle{IsTwocolumn}{
\begin{align} \label{eq:DfV_complete}
D_{(\ref{eq:lin_RFDE})}^+V(x_t)&=-x^\top\! (t) \hspace{0.5pt}  Q_0^{}  x(t)-x^\top \!(t-\delay)  \hspace{0.5pt}Q_1^{}  x(t-\delay) \nonumber \\
&\quad -\int_{-\delay}^0 x^\top\!(t+\theta)\hspace{0.5pt} Q_2^{}  x(t+\theta)\,\mathrm d \theta,\\[-1.5em]  \nonumber
\end{align}
}
{
\begin{align} \label{eq:DfV_complete}
D_{(\ref{eq:lin_RFDE})}^+V(x_t)&=-x^\top\! (t) \hspace{0.5pt}  Q_0^{}  x(t)-x^\top \!(t-\delay)  \hspace{0.5pt}Q_1^{}  x(t-\delay) 
-\int_{-\delay}^0 x^\top\!(t+\theta)\hspace{0.5pt} Q_2^{}  x(t+\theta)\,\mathrm d \theta,\\[-1.5em]  \nonumber
\end{align}
}
with freely chosen $Q_0, Q_1\succ 0_{n\times n}$, $Q_2\succeq 0_{n\times n}$. This derivative is accomplished by so-called complete-type (if $Q_{0,1,2}\succ 0_{n\times n}$) or related LK functionals 
\cite{Kharitonov.2003}, \cite[Thm.\ 2.11] {Kharitonov.2013}. Their determination is far more elaborate than the simple Lyapunov equation for ODEs:   
the known formula   for the solution of (\ref{eq:DfV_complete}) 
\begin{align} \label{eq:complete_LK_intro}
\iftoggle{IsTwocolumn}{&}{}
V(x_t) 
=
\iftoggle{IsTwocolumn}{}{\;&}
\iftoggle{IsTwocolumn}{\nonumber 
\\
&}{}
x^\top\!(t) \LM(0;\tilde Q) \,x(t)
+
2 \int_{-\delay}^0 x^\top\!(t) \LM(-\delay-\eta;\tilde Q)A_1 \hspace{0.5pt}x(t+\eta)\,\mathrm d \eta
\nonumber \\
&+
\int_{-\delay}^0\int_{-\delay}^0 x^\top\!(t+\xi) A_1^\top \LM(\xi-\eta;\tilde Q) A_1 \hspace{0.5pt}x(t+\eta)\,\mathrm d \eta\,\mathrm d\xi
\nonumber  \\
&+ 
\int_{-\delay}^0 x^\top\!(t+\eta) \big[Q_1+(\delay+\eta)Q_2\big] \hspace{0.5pt}x(t+\eta)\,\mathrm d \eta  
\end{align}
 requires 
the so-called delay Lyapunov matrix function\footnote{$\Psi(s;\tilde Q)$ is commonly denoted by $U(s)$ in the literature \label{fn:Psi}} 
$\LM(\,\cdot\,; \tilde Q)\colon[-\delay,\delay]\to \mathbb R^{n\times n}$
associated with $\tilde Q=Q_0+Q_1+\delay Q_2$. This matrix-valued function $\Psi$ is defined via a matrix-valued time-delayed boundary-value problem \cite[Def.~2.5]{Kharitonov.2013} 
that first has to be solved semi-analytically or numerically. The lower bound of interest on $V(x_t)$, 
 e.g., needed in an estimation of the domain of attraction \cite[Thm.~1]{MelchorAguilar.2007}, 
is described by 
\begin{align} \label{eq:pos_def_LK}
k_1\|x(t)\|^2\leq V(x_t). 
\end{align}
In contrast to the ODE case, where the minimum eigenvalue of $P$ gives the best possible coefficient $k_1$, nothing is reported about the conservativity of known formulae \cite{Kharitonov.2013,MelchorAguilar.2007} for 
(\ref{eq:pos_def_LK}).

\subsection{Objectives and Related Results}

 In light of the previous section, we intend to benefit from the enormous simplification that comes along with the treatment of ODEs in contrast to RFDEs. To this end, we use schemes of ODEs that approximate the RFDE. Based on these, the 
paper aims to provide a  new numerical approach to complete-type or related LK functionals which only requires to solve (a sequence of) Lyapunov equations. 
Moreover, a main objective is to get an improved   
coefficient~$k_1$ in  
 (\ref{eq:pos_def_LK}). 
Additionally, we hope to make the Lyapunov-Krasovskii theory more transparent, by interpreting the results in terms of Lyapunov-Rumyantsev partial stability of the approximating ODE.

Numerical approaches to complete-type and related LK functionals are a recent field of research. However, existing results either rely on the knowledge of the 
delay Lyapunov matrix function\footnotemark[\getrefnumber{fn:Psi}] 
$\Psi$,  
\cite{Gomez.2019b,Egorov.2017,Mondie.2022,Medvedeva.2015,Medvedeva.2013,Cuvas.2016,Bajodek.2023,Egorov.2014}, or they aim to determine 
$\Psi$, 
\cite{Egorov.2018,Michiels.2019,GarciaLozano.2004,Huesca.2009}. In contrast, the procedure in the present paper directly leads to an approximation of the overall LK functional (\ref{eq:complete_LK_intro}). 
Our main focus is not to provide a stability criterion, but, as outlined in Section \ref{sec:motivation}, we are interested in the functional itself and, in particular, in its lower bound (\ref{eq:pos_def_LK}). 

 We are going to use so-called \textit{discretization of the infinitesimal generator} approaches, which are well-established for numerical eigenvalue calculations  \cite{Breda.2015}. 
These approaches provide an ODE approximation of the RFDE. To analyze that ODE is also the core idea, e.g., in \cite{Breda.2016,Diekmann.2019,Wolff.2021}. 
The involved ODE can be obtained by various methods. 
We resort exemplarily to the Chebyshev collocation method, also known as pseudospectral discretization \cite{Breda.2005}, and to the Legendre tau method \cite{Ito.1986}. 

Even in the context of more general 
LK functionals,  a discretization of the RFDE in whatever form seems to be rarely  
considered in the literature. An early existence proof for quadratic LK functionals \cite{Infante.1978}, as well as a recent approach to so-called safety functionals \cite{Kiss.2021}, also employ discretizations. 
These, however, do not lead to ODEs, but to difference equations (so-called \textit{discretization of the solution operator} approaches \cite{Breda.2006b}).    
Moreover, in \cite{Bajodek.2020}, a discretization occurs in a proof of a linear-matrix-inequality 
stability criterion.

The core of the approach in the present paper is a Lyapunov equation from the ODE system matrix. The system matrices from both used discretization schemes are already known to give applicable Lyapunov, or, more generally, Riccati equation solutions. 
Concerning Chebyshev collocation, the resulting system matrix   has 
successfully been employed 
for Lyapunov equations in the context of $H_2$-norm computations \cite{Jarlebring.2011b,Vanbiervliet.2011,Michiels.2019}, 
where the delay Lyapunov matrix  $\Psi(0;\tilde Q)$ at $s=0$ is of interest. 
Further calculations are mentioned in \cite{Michiels.2019} to obtain, at least under the assumption of an exponentially stable RFDE equilibrium, 
the matrix-valued function $\Psi$ for the LK functional formula (\ref{eq:complete_LK_intro}). 
The Lyapunov equation is a common element with the present paper, but
only a submatrix of the Lyapunov equation solution is used in \cite[Prop.~2.1]{Michiels.2019}, 
the product with a  matrix exponential is required for any value of $s$ in $\Psi(s;\tilde Q)$, and the integral expressions in (\ref{eq:complete_LK_intro}) 
still would have to be evaluated to obtain a LK functional value. 
Concerning 
Legendre tau, 
the 
system matrix  (respectively a similar matrix) has already successfully  been used for algebraic  Riccati equations in the context of optimal control  \cite{Ito.1987}.

{\itshape Structure.} The paper is organized as follows. Sec.~\ref{sec:ODE_Approx} describes the numerical approach, and Sec.~\ref{sec:quadLowerBound} gives the formula for the quadratic lower bound, which is applied to an example in Sec.~\ref{sec:Example}. 
In Sec.~\ref{sec:role_of_part_stab},  we interpret the approach in terms of partial stability of the approximating ODE. 
Finally, Sec.~\ref{sec:Convergence} addresses convergence, before  Sec.\ \ref{sec:Conclusion} concludes the paper.

{\itshape Notation.}
The space of continuous  $\mathbb R^n$-valued functions on the interval $[a,b]$ is denoted by $C([a,b],\mathbb R^n)$, in short $C$, and square integrable functions by $L_2([a,b],\mathbb R^n)$ or $L_2$. 
We write $(w_k)_{k\in\mathcal I}$ for a vector with entries $w_k$, e.g., $(w_k)_{k\in\{0,\ldots,N\}}=[w_0,\ldots, w_N]^\top$, or  $(w_k)_k$ if the index set is clear from the context. Similar holds for matrices.
The set of eigenvalues of 
$A\in \mathbb R^{n\times n}$ is $\sigma(A)$, and 
$A$ is said to be Hurwitz if all eigenvalues have negative real parts. 
Moreover, $Q\succ 0_{n\times n}$ $(Q\succeq 0_{n\times n})$ denotes positive (semi)definiteness of $Q\in\mathbb R^{n\times n}$, implicitly requiring that $Q=Q^\top$. 
The zero vector in $\mathbb R^n$ is $0_n$, the vector-valued zero function on $[a,b]$ is $0_{n_{[a,b]}}$, the $m\times n$ zero matrix   $0_{m\times n}$, and the identity matrix in $\mathbb R^{n\times n}$ is $I_n$. 
Given $x\in \mathbb R^n$, we write $\|x\|_2$ for the Euclidean norm, whereas $\|x\|$ can be any arbitrary norm in $\mathbb R^n$. 
The Kronecker product of two matrices $A$ and $B$ is $A\otimes B$, and $A^-$ denotes a generalized inverse. 
To emphasize the structure of a block matrix, e.g., $A=[A_1\;A_2]$, with differently sized submatrices, $A_1\in\mathbb R^{n\times nN}$, $A_2\in \mathbb R^{n\times n}$, we write $A=[\;{\rule[.5ex]{2.5ex}{0.3pt}} A_1 {\rule[.5ex]{2.5ex}{0.3pt}}\;\;\; A_2\;]$. We use $\stackrel !=$ to  mark  a requirement, and $\stackrel{(\ldots)}=$ if the relation is explained by $(\ldots)$. 
The set of class-K functions is defined by 
$
\mathcal K = \{\kappa\in C([0,\infty),\mathbb R_{\geq 0}): \kappa(0)=0, \text{ strictly increasing}\}$. For the formal definition of  
$D_{(\mathrm{eq})}^+V$, see, e.g., \cite[Sec.~5.2]{Hale.1993}. 

\section{The numerical approach} \label{sec:ODE_Approx}
\subsection{ODE-Approximation Schemes of RFDEs}
Given a continuous initial function $x_0\in C([-\delay,0],\mathbb R^n)$,  the state  $x_t\in C([-\delay,0],\mathbb R^n)$ of the RFDE at time $t\geq 0$ 
 is  
defined by $x_t(\theta)=x(t+\theta)$, $\theta\in[-\delay,0]$.  Thus, it  
 represents the solution segment on $[t-\delay,t]$, cf.\ Fig.~\ref{fig:sol_segment}/\ref{fig:sol_state}. An ODE approximation has to 
address a finite-dimensional state vector 
instead. 
In the simplest case, this state vector  
$y(t)$ 
at time $t$ 
approximates 
 the values of 
 the segment 
$x_t$  in  
$N\!+\!1$ ordered points $\tilde \theta_0=-\delay,\;\ldots,\; \tilde \theta_{N}=0$,  
\begin{align} \label{eq:y_vec}
\begin{bmatrix}
x(t-\delay)\\
x(t+\tilde \theta_1)
\\
\vdots 
\\
x(t+\tilde \theta_{\!N-1})\!
\\
x(t)
\end{bmatrix}
=
\begin{bmatrix}
x_t(-\delay)\\
x_t(\tilde \theta_1)
\\
\vdots 
\\
x_t(\tilde \theta_{\!N-1})\!
\\
x_t(0)
\end{bmatrix}
\approx
\underbrace{\begin{bmatrix}
y^0(t)\\
y^1(t)
\\
\vdots 
\\
y^{N-1}(t)
\\
y^N(t)
\end{bmatrix}}_{y(t)\in \mathbb R^{n(N+1)}}
=:\!
\left[\begin{array}{@{}c}
z^0(t)\\
z^1(t)
\\
\vdots 
\\
z^{N-1}(t)\!
\\[0.2em]
\hline
\\[-0.8em]
\hat x(t)
\end{array}
\right]\!.
\nonumber \\[-1.2em]
\end{align}
Henceforth, upper indices  $k\in\{0,\ldots,N\}$ address vector-valued components $y^k\hspace{-0.5pt}(t)\hspace{-1pt}\in \hspace{-0.5pt}\mathbb R^n$.
Whenever the special interest in 
$y^N$ shall be emphasized, we use the indicated decomposition $y=[z^\top\!,\hat x^\top]^\top\!$. 
For 
$\tilde \theta_k$ in  (\ref{eq:y_vec}), 
a non-equidistant grid  
\begin{align} \label{eq:ChebNodes}
\tilde \theta_k
=\tfrac \delay 2 (\tilde \vartheta_{k}-1),
\qquad \text{with} \quad
\tilde \vartheta_k
=-\cos(\tfrac k N \pi),
\end{align}
$k\!\in\!\{0,\ldots,N\}$,    
built from 
shifting and scaling  classical  
Chebyshev nodes\footnote{also called Gauss–Lobatto Chebyshev nodes (cf.\ Table \ref{tab:polynomialApprox}) or  Chebyshev points of the second kind (despite of referring to extrema of the  `Chebyshev polynomials of the first kind') or endpoints-and-extrema Chebyshev nodes.} 
$\tilde \vartheta_k\in[-1,1]$  to $\tilde \theta_k\in [-\delay,0]$,  
 has proven to be advantageous \cite{Trefethen.2000}. 
The latter is 
also at the core of the open-source Matlab toolbox Chebfun by Trefethen and co-workers \cite{Driscoll.2014}, from which we can benefit in the implementations. 
 
It remains to find the ODE 
 \begin{align} \label{eq:y_ODE}
\dot y(t)=A_{\ind y}\, y(t),
\end{align}
$A_{\ind y}\in \mathbb R^{n(N+1)\times n(N+1)}$, that describes the dynamics of $y$. To this end, we use exemplarily the Chebyshev collocation method and the Legendre tau method combined with a change of basis. The resulting system matrices $A_y$ are given by (\ref{eq:A_y_C}) and (\ref{eq:A_y_L}) in the appendix.
Fig.~\ref{fig:sol_ODE} shows how a solution of (\ref{eq:y_ODE})  looks like, provided the initial condition  $y(0)$, 
given by the blue points, is a discretization
 of the initial function $x_0\in C([-\delay,0],\mathbb R^n)$, cf.\ (\ref{eq:phi2y}) with $\phi=x_0$.

\begin{figure}
		\centering
\subfloat[
RFDE solution, state $x_t$ as solution segment
		\label{fig:sol_segment}
	]
	{
	\includegraphics{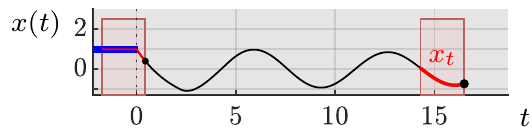} 
}	
\\
\subfloat[
evolution of the RFDE state $x_t$
		\label{fig:sol_state}
	]
	{
\includegraphics{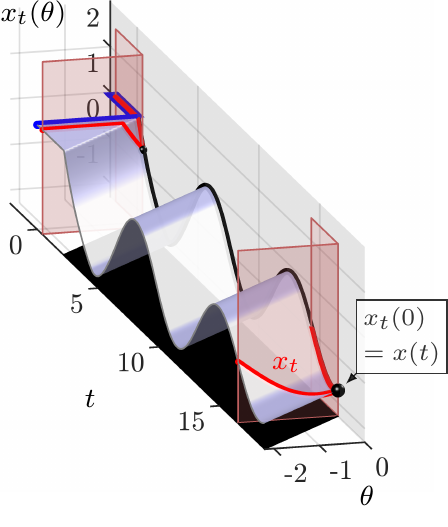}
}	
\iftoggle{IsTwocolumn}{}{\quad}
\subfloat[
components $y^k(t)$ of the ODE solution ($N=16$, $A_y$ 
from (\ref{eq:A_y_C}))
		\label{fig:sol_ODE}
	]
		{\hspace{-1.6cm}
		
		\includegraphics[]{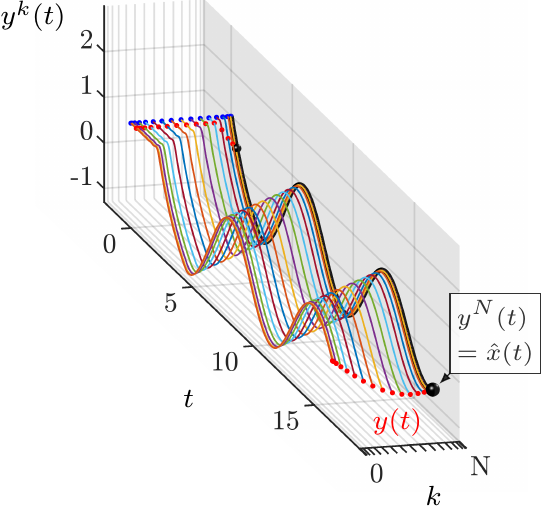}
\hspace{-1em}
}	
	\caption{Solution of $\dot x(t)=-0.5x(t)-x(t-2.2)$ for the initial function $x_0(\theta)\equiv 1$. 
	}
	\label{fig:sol}
\end{figure}

\subsection{An Approximation Scheme for the LK Functional} \label{sec:Partial_Lyap_Eq}

We are going to set up a Lyapunov function 
$V_y\colon\mathbb R^{n(N+1)}\to \mathbb R$
for the approximating ODE (\ref{eq:y_ODE}) (in fact, a partial Lyapunov function, see Sec.~\ref{sec:role_of_part_stab}). 
To this end, we make the quadratic ansatz 
\begin{align}\label{eq:Vy}
V_{\ind y}(y)=y^\top P_{\ind y}\, y, 
\end{align}
with $P_{\ind y}=P_{\ind y}^\top\in \mathbb R^{n(N+1)\times n(N+1)}$ to be determined. 
The derivative of $V_{\ind y}$ 
along solutions shall be 
$-y^\top Q_{\ind y} y$ with a  prescribed symmetric 
matrix $Q_{\ind y}$
\begin{align} \label{eq:DV_Qy}
D_{(\ref{eq:y_ODE})}^+V_{\ind y}(y)=y^\top (P_{\ind y}A_{\ind y}+A_{\ind y}^\top P_{\ind y}) y \stackrel !=-y^\top Q_{\ind y} y,
\end{align}
$ \forall y\in \mathbb R^{n(N+1)}$.
Thus, the unknown matrix $P_{\ind y}$ is obtained by solving the Lyapunov equation 
\begin{align}\label{eq:LyapEq_y}
P_{\ind y}A_{\ind y}+A_{\ind y}^\top P_{\ind y}=-Q_{\ind y}.
\end{align}  
See Appendix \ref{sec:LegCoordinates} for a description in Legendre coordinates (indicated by a subscript $\zeta$ at the matrices). 
We construct the right-hand side of (\ref{eq:DV_Qy}) 
according to a discretization of the right-hand side of  (\ref{eq:DfV_complete}) with freely chosen matrices $Q_0,Q_1\succ 0_{n\times n}$, $Q_2\succeq 0_{n \times n}$. Hence, a straightforward choice of 
$Q_{\ind y}$ in 
(\ref{eq:LyapEq_y}) becomes visible from 
\begin{align}
\iftoggle{IsTwocolumn}{&}{}
D_{(\ref{eq:y_ODE})}^+V_{\ind y}(y) 
\iftoggle{IsTwocolumn}{}{&}
\stackrel ! = 
- (y^N) \!^\top Q_0 y^N \!
-(y^0)\!^\top Q_1 y^0
- \!\sum_{k=0}^N (y^k )\!^\top Q_2 y^k w_k
\nonumber\\
&=-y^\top \!\!
\left(
\left[\begin{array}{@{}c@{}c@{}c@{}c@{}c@{}} 
Q_1  &&&&\\ 
& 0_{n\times n} \!&&& \\ 
&&\ddots && \\ 
&&& 0_{n\times n} \! &\\ 
&&&& Q_0
\end{array}\right]
+
\left[\begin{array}{@{}c@{}c@{}c@{}c@{}c@{}} 
w_0 Q_2  &&&&\\ 
&\!\!\! w_1 Q_2   &&& 
\\ 
&& \ddots && \\ 
\\ 
&&&& w_N Q_2 
\end{array}\right]
\right)
y
 \nonumber 
\\
&=:-y^\top Q_{\ind y} y\,,
\label{eq:Q_quadForm}
\end{align}
where $w_k\in \mathbb R$ are integration weights, see Appendix \ref{sec:num_int}. 
Sec.~\ref{sec:splitting} will present discretization-scheme-dependent modifications of $Q_y$ that aim at improved convergence properties. 

Altogether, we 
solve a discretization of the original problem (\ref{eq:DfV_complete}), and thus $V_y(y)$ in (\ref{eq:Vy}) is intended to be an approximation of the LK functional $V(\phi)$. Convergence aspects will be addressed in Sec.~\ref{sec:Convergence}. 
Hence, given a prescribed argument $\phi\in C([-\delay,0],\mathbb R^n)$, which might be $\phi=x_t$ for some $t\geq 0$, or, 
without loss of generality, $\phi=x_0$ at $t=0$, we 
can obtain a numerical approximation for the evaluation $V(\phi)$. 
To this end, the argument $y$ in $V_y(y)$ must be chosen correspondingly. 
Such a discretization $y$ of $\phi$ can be  obtained by evaluating 
the vector-valued function 
$\phi$ 
at the gridpoints (\ref{eq:ChebNodes}) and 
stacking
these $(N+1)$ vectors in    
\begin{align} \label{eq:phi2y}
y
=
\left[\begin{array}{c}
\\[-0.25em]
\vert
\\
z
\\
\vert
\\[0.75em]
\hat x
\end{array}
\right]
=
\begin{bmatrix}
\phi(-\delay)\\
\phi(\tilde \theta_1)
\\
\vdots 
\\
\phi(\tilde \theta_{\!N-1})\!
\\
\phi(0)
\end{bmatrix}. 
\end{align}
Strictly speaking, (\ref{eq:phi2y}) is the  interpolatory discretization presupposed in the Chebyshev collocation method. 
If $\phi$ is a polynomial of order $N$ or less, (\ref{eq:phi2y}) also agrees with  
the coordinate transform (\ref{eq:T_yc}) of the 
discretization 
in the Legendre tau method 
(\ref{eq:zeta_discretization}), but otherwise the latter might give a slightly deviating vector  $y$ (pointwise evaluations of the approximating polynomial).   

To sum up, 
we only have to solve the Lyapunov equation (\ref{eq:LyapEq_y}), to obtain the approximation $V(\phi)\approx V_y(y)$.

\subsection{Existence, Uniqueness, and Non-Negativity}\label{sec:ExistenceUniqueness}

Note that $Q_{\ind y}$ in (\ref{eq:Q_quadForm}) is a positive semidefinite, but not necessarily positive definite matrix. 
Let us revisit some properties of the Lyapunov equation (\ref{eq:LyapEq_y}) in this rather uncommon semidefinite case, without further  assumptions on the involved matrices. 
See \cite[p.\ 284]{Hinrichsen.2005}, and \cite[Thm.\ 1]{Carlson.1962} for Lemma \ref{lem:LyapEquationProperties}c.
\begin{lemma}\label{lem:LyapEquationProperties}
Consider $PA+A^\top P=-Q$,\; $A,Q\in \mathbb R^{\nu\times \nu}$. 
\\
(a) If $\sigma(A)\cap(-\sigma(A))=\emptyset$, then  a unique solution $P$ exists.
\\
(b) If $Q=Q^\top$ and $P$ is a solution, then $P^\top$ is also a solution. If, additionally, (a) holds, then $P=P^\top$.
\\
(c) If $Q\succeq 0_{\nu\times \nu}$, $P=P^\top$, and $i_0(A)=0$, then $i_+(P)\leq i_-(A)$ and $i_-(P)\leq i_+(A)$, where 
$i_{-,0,+}$ are 
the numbers of eigenvalues with negative, zero, and positive real parts.
\end{lemma}
\begin{remark}\label{rem:LyapCond}
Existence of the 
LK functional $V$ 
in (\ref{eq:DfV_complete})
is analogously ensured by the time-delay counterpart of Lemma \ref{lem:LyapEquationProperties}a, 
the so-called Lyapunov condition \cite[Def.~2.6]{Kharitonov.2013}.  
\end{remark}
\begin{proposition}\label{prop:properties}
Let $Q_y\succeq 0_{n(N+1)\times n(N+1)}$ be given. If 
$A_y$ is Hurwitz, then there exists a unique solution $P_y$ in (\ref{eq:LyapEq_y}). Moreover, $P_y=P_y^\top$ is positive semidefinite. 
\end{proposition}
\begin{proof}
Lemma \ref{lem:LyapEquationProperties}a with $\sigma(A)\subset\mathbb C^-$, Lemma \ref{lem:LyapEquationProperties}b, and Lemma \ref{lem:LyapEquationProperties}c with $i_0(A)=i_+(A)=0$.
\end{proof}
Consequently, if the zero equilibrium of the ODE approximation (\ref{eq:y_ODE}) is asymptotically stable, and $D_{(\ref{eq:y_ODE})}^+V_y(y)$ is chosen according to  (\ref{eq:Q_quadForm}) and thus nonpositive, 
then 
existence, uniqueness, and nonnegativity of $V_y(y)$ in Sec.~\ref{sec:Partial_Lyap_Eq} are ensured.

\subsection{Structure of the Result} 

 To get an impression of how the Lyapunov equation solution $P_y$ looks like, we  consider an example with $n=1$. As will be demonstrated, only little implementation effort is required.

\begin{figure} 
		\centering

\includegraphics{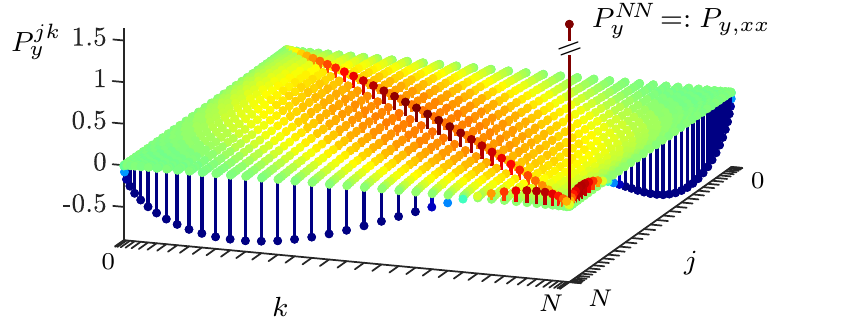} 

	\caption{Entries of the matrix $P_{\ind y}$ in Example \ref{exmp:P_structure} ($N=40$).
	}
	\label{fig:P_structure}
\end{figure}

\begin{example} \label{exmp:P_structure}
Let $\dot x(t)=-0.5x(t)-x(t-2.2)$  and $Q_0=Q_1=1$, $Q_2=0$ in (\ref{eq:Q_quadForm}). 
We get the solution $P_y$ of (\ref{eq:LyapEq_y}) via\footnote{
If $Q_2\neq 0_{n\times n}$ and $A_y=A_y^L$, then 
\texttt{Tcy'*Qc2*Tcy}  from (\ref{eq:Qzeta2}) is added to $\texttt{Q}$, 
with  \texttt{Qc2=kron(delay*diag([1./(2*(0:N-1)+1),1]),Q2)}.  
}  \vspace{0.2em}
{\small\begin{verbatim} 
 Q=blkdiag(Q1,zeros(n*(N-1)),Q0); P=lyap(A',Q);
\end{verbatim}}\noindent
in Matlab, 
provided $A_y$ is assigned to $\texttt{A}$ (see Rem.~\ref{rem:Matlab_AyC} or Rem.~\ref{rem:Matlab_AyL} in the appendix). 
The structure of $P_{\ind y}$ for $N=40$ is depicted in Fig.~\ref{fig:P_structure}. It  stems from the Legendre tau method, i.e., $A_{\ind y}=A_{\ind y}^L$ from (\ref{eq:A_y_L})  is used in the Lyapunov equation (or, equivalently, $A_{\zeta}=A_{\zeta}^L$ from (\ref{eq:A_zeta_L}) in the Lyapunov equation from Appendix \ref{sec:LegCoordinates}). 
However,  Chebyshev collocation with $A_{\ind y}=A_{\ind y}^C$ from (\ref{eq:A_y_C}) 
	gives 
almost the same 
 picture of $P_{\ind y}$.
\end{example}

In Fig.\ \ref{fig:P_structure}, the combs on the last column, the last row, and the diagonal as well as the striking right lower element of the matrix $P_y$    are also existent with a refined grid. 
Thus, $V_{\ind y}(y)=y^\top P_{\ind y}\, y=\sum_{j=0}^N \sum_{k=0}^N (y^j)^\top P_{\ind y}^{jk} y^k$ is not the discretized version of a Lebesque integral $\int_{-\delay}^0\int_{-\delay}^0 \phi^\top\!(\xi)\, P(\xi,\theta) \,\phi(\theta)\,\mathrm d \theta \mathrm d \xi$. 
Instead, the combs suggest that 
\begin{align} \label{eq:V_y_matrix}
\iftoggle{IsTwocolumn}{&}{}
V_{\ind y}(y)
\iftoggle{IsTwocolumn}{}{&}
=y^\top P_{\ind y} y
=
\begin{bmatrix} 
\rule[-1ex]{0.3pt}{1.5ex} \\
z \vphantom{\hat P^\top_x}\\   
\rule[0ex]{0.3pt}{1.5ex} \\ 
\hat x \vphantom{\hat P^\top_x}
\end{bmatrix}^\top 
\begin{bmatrix}
 P_{{\ind y},zz} 
& 
\begin{matrix}  
\rule[-0ex]{0.3pt}{1.5ex} \vphantom{\rule[-0.5ex]{0.3pt}{1.5ex}  }\\ 
 P_{{\ind y},xz}^\top \\    
\rule[-0.2ex]{0.3pt}{1.5ex}  \vphantom{\rule[-0ex]{0.3pt}{1.5ex}}
\end{matrix}
\\[1.5em]
{\rule[.5ex]{1.5ex}{0.3pt}}  
\; P_{{\ind y},xz}\,
{\rule[.5ex]{1.5ex}{0.3pt}} 
\;
&
P_{{\ind y},xx}
\end{bmatrix}
\begin{bmatrix} 
\rule[-1ex]{0.3pt}{1.5ex} \\
z \vphantom{\hat P^\top_x}\\   
\rule[0ex]{0.3pt}{1.5ex} \\ 
\hat x \vphantom{\hat P^\top_x}
\end{bmatrix}
\\
&=
\hat x^\top P_{{\ind y},xx} \hat x
+ 
2 \sum_{k=0}^{N-1} \hat x^\top P_{{\ind y}, xz}^k z^k 
+ 
\sum_{j=0}^{N-1} \sum_{\substack{k=0\\k\neq j}}^{N-1} 
(z^j)^\top  P_{{\ind y},zz}^{jk} z^k 
\iftoggle{IsTwocolumn}{
\nonumber\\[-1em]
&\quad}{}
+
\sum_{k=0}^{N-1} (z^k)^\top P_{{\ind y},zz}^{kk} z^k
\label{eq:V_y_sums}
\end{align}
describes, through the (discrete $\hookrightarrow$ continuous) correspondences indicated by (\ref{eq:phi2y}) and by $k$ vs. $\theta$ in  Fig.~\ref{fig:sol}  
\begin{align*}
z^k&=\phi(\tilde\theta_k), \;k\!\hspace{0.5pt}\in \!\{0,\ldots, N\!-\!1\} 
\quad 
 &\hookrightarrow  
\quad 
&\phi(\theta),\;\theta\in[-\delay,0),
\\
z^j&=\phi(\tilde\theta_j), \;j\!\hspace{0.5pt}\in \!\{0,\ldots, N\!-\!1\} 
\quad 
 &\hookrightarrow  
\quad 
&\phi(\xi),\;\xi\in[-\delay,0),
\\
\hat x&=\phi(0)& \hookrightarrow   \quad &\phi(0), 
\end{align*}
the discrete version of some 
\begin{subequations}\label{eq:complete_LK_with_kernels}
\begin{align} \label{eq:complete_LK}
V(\phi) 
&=
\phi^\top\!(0) \hspace{0.5pt}P_{\mathrm{xx}} \hspace{0.75pt}\phi(0)
+
2 \int_{-\delay}^0 \phi^\top\!(0) \hspace{0.5pt}P_{\mathrm{xz}}(\theta) \hspace{0.5pt}\phi(\theta)\,\mathrm d \theta
\iftoggle{IsTwocolumn}{\nonumber \\ &\quad}{}
+
\int_{-\delay}^0\int_{-\delay}^0 \phi^\top\!(\xi) P_{\mathrm{zz}}(\xi,\theta) \hspace{0.5pt}\phi(\theta)\,\mathrm d \theta\,\mathrm d\xi
\nonumber  \\&\quad
+ 
\int_{-\delay}^0 \phi^\top\!(\theta) \, P_{\mathrm{zz,diag}}(\theta) \,\phi(\theta)\,\mathrm d \theta . 
\end{align}
Note that the latter exactly reflects the known structure of complete-type and related LK functionals given in (\ref{eq:complete_LK_intro}).

\subsection{Validation via Numerical Integration}\label{sec:numInt}
To be more precise, the structure of complete-type and related LK functionals is the one in (\ref{eq:complete_LK}), and the kernel functions can be identified in (\ref{eq:complete_LK_intro}) as 
\begin{alignat}{5}
P_{\mathrm{zz}}(\xi,\theta)&=A_1^\top \LM(\xi-\theta;\tilde Q) A_1,
\;
 &\;P_{\mathrm{xz}}(\theta)&=\LM(-\delay-\theta;\tilde Q)A_1, 
\nonumber\\* 
P_{\mathrm{zz,diag}}(\theta) &=Q_1+(\delay+\theta)Q_2, 
&P_{\mathrm{xx}}&=\LM(0;\tilde Q).
\label{eq:P_in_completeTypeLK}
\end{alignat}
\end{subequations}
For the sake of validation, we also go the other way around and discretize  the known formula of $V(\phi)$ 
by interpolatory quadrature rules (cf.\ Table \ref{tab:polynomialApprox}).  That is, replacing the integrals in (\ref{eq:complete_LK})  by weighted sums from evaluations at the grid points. 
In Appendix \ref{sec:num_int}, we write the 
result as a quadratic form
\begin{align}\label{eq:P_quad_Clenshaw_Curtis_}
V(\phi)\approx y^\top \!P_{\ind{y}}^{\,\mathrm{quad} } y  
\end{align} 
like (\ref{eq:V_y_matrix}).
Taking  for $\Psi$ in (\ref{eq:P_in_completeTypeLK}) the semi-analytical solution approach from \cite{GarciaLozano.2004}, 
the picture of  the resulting $P_{\ind{y}}^{\,\mathrm{quad} }$ 
for Example \ref{exmp:P_structure} is indeed hardly distinguishable from Fig.~\ref{fig:P_structure}. See Sec.~\ref{sec:Example} for further numerical comparisons. 
\begin{remark}
Both the ODE-based approach from Sec.~\ref{sec:Partial_Lyap_Eq} and the numerical-integration-based approach from Sec.~\ref{sec:numInt} provide an approximation $V_y(y)=y^\top P_y y$. The former seeks for an approximative solution of the defining equation (\ref{eq:DfV_complete}).      
In contrast, the latter already starts with the exact knowledge of the LK functional (\ref{eq:complete_LK_intro}), presupposing knowledge of $\Psi$, and only has to describe a 
discretization thereof. In so far, 
the numerical-integration-based approach 
is related to discretizations of the known $V(\phi)$ 
already 
proposed in the literature 
-- be it based on piecewise cubic polynomials that approximate $\phi$ \cite{Medvedeva.2015,Medvedeva.2013} 
or, recently, on a Legendre series truncation of $\phi$ 
\cite{Bajodek.2023} 
(also used in 
\cite{Bajodek.2024,Bajodek.2022e}), 
or a certain fundamental-matrix-dependent  discontinuous approximation of $\phi$   \cite{Gomez.2019b,Egorov.2017,Egorov.2014}.
With the exception of the latter approach (which addresses 
zero $Q_0$ and $Q_2$), integral terms 
with $\Psi$ 
must still be evaluated. 
To our best knowledge, 
applying interpolatory 
quadrature rules (cf.\ Table \ref{tab:polynomialApprox})  to 
(\ref{eq:complete_LK_with_kernels})  
has not yet been considered. 
\end{remark}

\section{The Quadratic Lower Bound} \label{sec:quadLowerBound}
As a result of the preceding section, we have an approximation of the LK functional. 
However, in applications, we also need the quadratic lower bound (\ref{eq:pos_def_LK}). 
If $Q_0,Q_1\succ 0_{n\times n}, Q_2\succeq 0_{n\times n}$,   
existence of a non-zero\footnote{\label{fn:zerok1}In contrast to quadratic forms from finite-dimensional matrices, in infinite dimensions coercivity 
of the associated bilinear form (existence of a quadratic lower bound) is a stronger concept than positive definiteness (positivity for any nonzero element). 
The same holds for 
the partial 
concepts. 
Consequently, despite of $V_y$ being partially positive definite w.r.t.\ $\hat x$ (Def.~\ref{def:partPosDef}),  
the largest possible coefficient in (\ref{eq:k1_Vy}) as $N\to\infty$ could become $k_1\to 0$, 
cf.\ Rem.~\ref{eq:cubic_lower_bound}. 
}
coefficient $k_1>0$ 
is proven in \cite[Lem.\ 2.10]{Kharitonov.2013}, given the RFDE equilibrium is exponentially\footnotemark[\getrefnumber{fn:UAS_UES}]
 stable.
In a discrete version for the approximation $V_y$,
the 
bound (\ref{eq:pos_def_LK}) becomes $\forall y=[z^\top, \hat x ^\top ]^\top, 
 z\in \mathbb R^{nN},\hat x\in \mathbb R^n:$
\begin{align} \label{eq:k1_Vy}
k_1 \| \hat x\|_2^2 
\leq 
 V_{\ind{y}}(y). 
\end{align}
Since solely  $\hat x=y^N$ is considered, (\ref{eq:k1_Vy}) does not refer to the common $\lambda_{\min}(P_y)\| y\|_2^2 \leq  V_{\ind{y}}(y)$ 
mentioned in the introduction. 
Why this discrete version of (\ref{eq:pos_def_LK}) still also makes sense in a Lyapunov analysis of the approximating ODE, will be 
 explained in Sec.~\ref{sec:role_of_part_stab}.

The main contribution of the present section, Lemma \ref{lem:part_pos_def_bound_Schur}, immediately leads to the searched bound (\ref{eq:k1_Vy}) in Thm.\ \ref{thm:k1}. For the sake of readability, we consider a general positive semidefinite matrix $P$ with a left upper submatrix $Z$, instead of $P_y$ and $P_{y,zz}$ introduced in (\ref{eq:V_y_matrix}).  
The lemma is based on the generalized Schur complement (\ref{eq:SchurComplement}), cf.\ \cite{Horn.2005}, where $Z^-$ is a generalized matrix inverse of $Z$, e.g., the Moore-Penrose inverse. If $Z$ is nonsingular, then $Z^-=Z^{-1}$. 
\begin{lemma}\label{lem:part_pos_def_bound_Schur}
Let $P=[\begin{smallmatrix} Z & B \\ B^\top & X \end{smallmatrix}]$ with $Z=Z^\top\in \mathbb R^{p\times p}, B\in \mathbb R^{p\times n}, X=X^\top \in \mathbb R^{n\times n}$. If $P$ is positive semidefinite, then
\begin{align}\label{eq:ABC_Schur}
\min_{\substack{z\in \mathbb R^p\\ x\in \mathbb R^n\setminus\{0_n\}}}\!
\frac{1}{\|x\|_2^2}
\begin{bmatrix}
z \\ x
\end{bmatrix}^{\!\top} \!\!
\begin{bmatrix}Z & B \\ B^\top & X \end{bmatrix}
\begin{bmatrix}
z \\ x
\end{bmatrix}
&=
\lambda_{\min}(P/Z),
\\
\text{where}\quad 
\label{eq:SchurComplement}
P/Z 
&= X-B^\top Z^- B. 
\end{align}
The minimum is attained by 
$ [\begin{smallmatrix}
z \\ x
\end{smallmatrix}]
= 
[\begin{smallmatrix}
-Z^-B v \\ v
\end{smallmatrix}]
$, 
with $v$ being an eigenvector in 
$ (P/Z) \, v = v\, \lambda_{\min}(P/Z) $.
\end{lemma}
\begin{proof}
Let us replace $z$ by $w:=z+Z^-B\,x$, 
which 
amounts to 
the coordinate transformation  
\begin{align}
\begin{bmatrix}
z \\ x
\end{bmatrix}
=
\begin{bmatrix}
I_p & -Z^-B \\ 0_{n\times p} & I_n
\end{bmatrix}
\begin{bmatrix}
w \\ x
\end{bmatrix}=:T_{\ind {yq}} \begin{bmatrix}
w \\ x
\end{bmatrix}.
\end{align}
We arrive at the so-called 
generalized Aitken 
block-diagonalization of $P$ in   
\begin{align}
\iftoggle{IsTwocolumn}{&}{}
\begin{bmatrix}
z \\ x
\end{bmatrix}^\top
\begin{bmatrix}Z & B \\ B^\top & X \end{bmatrix}
\begin{bmatrix}
z \\ x
\end{bmatrix}
\iftoggle{IsTwocolumn}{}{&}=
\begin{bmatrix}
w \\ x
\end{bmatrix}^\top
T_{\ind {yq}}^\top 
\begin{bmatrix} Z & B \\ B^\top & X \end{bmatrix}
T_{\ind {yq}}
\begin{bmatrix}
w \\ x
\end{bmatrix}
\nonumber\\
\iftoggle{IsTwocolumn}{&\hspace{1.3cm}}{&}
=
\begin{bmatrix}
w \\ x
\end{bmatrix}^\top\!
\begin{bmatrix}
Z & -ZZ^-B+B \\
\;-B^\top Z^-Z+B^\top \; & \;\;X-B^\top Z^- B\;
\end{bmatrix}
\begin{bmatrix}
w \\ x
\end{bmatrix}
\nonumber\\
\iftoggle{IsTwocolumn}{&\hspace{1.3cm}}{&}
=
 w^\top Z w + x^\top (P/Z) \,x,   
\label{eq:AitkenBlockDiag}
\end{align}
with the last step being based on    
$-ZZ^-B+B=0_{p\times n}$, which holds   
if $P$ is positive semidefinite \cite[Thm.\ 1.19]{Horn.2005}. The submatrix $Z$ of $P$ 
 is also positive semidefinite due to Cauchy's Interlacing Theorem, and thus 
(\ref{eq:AitkenBlockDiag}) is lower bounded by
\begin{align*}
\nonumber\\[-1.7em]
 w^\top Z \,w + x^\top (P/Z)\, x 
\;\geq \;
x^\top (P/Z) \,x 
\;\geq \;
\lambda_{\min}(P/Z) \|x\|_2^2. 
\end{align*}
The bound is attained for $w=0_p$ and 
$x=v$. 
\end{proof}
The following theorem is not only useful for the ODE-based approach from Sec.~\ref{sec:Partial_Lyap_Eq}. It is as well applicable to the numerical-integration-based results from Sec.~\ref{sec:numInt}. 

\begin{theorem}\label{thm:k1}
If $P_y$ in $V_y(y)=y^\top P_y y$ is positive semidefinite, then the largest possible coefficient in (\ref{eq:k1_Vy}) is 
\begin{align} \label{eq:kappa1_V_y}
k_1=
\lambda_{\min}(P_{\ind{y}}/P_{\ind{y,zz}}),
\end{align}
where $P_{\ind{y,zz}}$ denotes the left upper $nN\times nN$ submatrix of $P_y$ and $(\cdot / \cdot)$ is the generalized
 Schur complement (\ref{eq:SchurComplement}).
\end{theorem}
\begin{proof}
Lemma \ref{lem:part_pos_def_bound_Schur} 
applied 
to 
$P=P_y$ with $Z=P_{y,zz}$ as in (\ref{eq:V_y_matrix}).  
\end{proof}
Testing whether $P_y$ is positive semidefinite is not even required if $V_y$ originates from the ODE-based 
approach in Sec.~\ref{sec:Partial_Lyap_Eq}.
If $A_y$ is 
Hurwitz, 
the only thing to do is to evaluate  (\ref{eq:kappa1_V_y}).

\begin{corollary} \label{cor:lowerBound}
Let $V_y(y)=y^\top P_y y$, where $P_y$ is a solution of 
(\ref{eq:LyapEq_y}) for a given positive semidefinite matrix $Q_y$. 
If $A_y$ is Hurwitz, 
then (\ref{eq:k1_Vy}) holds with $k_1$ from  (\ref{eq:kappa1_V_y}).  
\end{corollary}
\begin{proof} 
By 
Prop.~\ref{prop:properties}, $P_{\ind y}$ is positive semidefinite. 
Consequently, Thm.~\ref{thm:k1} applies. 
\end{proof}
See Appendix \ref{sec:combined_coordinates} for an evaluation in other  coordinates.

\section{Example and Comparison}\label{sec:Example}
We compare the thus obtained bound with 
known quadratic lower bounds  (\ref{eq:pos_def_LK})
on the LK functional  
(\ref{eq:complete_LK_with_kernels}). 
These known formulae for the coefficient in 
$k_1\|x(t)\|_2^2 \leq V(x_t)$ 
are 
\begin{align*}
k_1&=\max \alpha 
\tag*{\cite[Lem.\ 2.10]{Kharitonov.2013}} 
\\*
&\hspace{1.5em} \mathrm{ s.t. }\;\; \left[\begin{smallmatrix} Q_0 & 0_{n\times n}\\ 0_{n\times n} & Q_1 \end{smallmatrix}\right]
+ \alpha  
\left[\begin{smallmatrix} A_0^\top+ A_0 &A_1\\ A_1^\top & 0_{n\times n} \end{smallmatrix}	\right]
\succeq 0_{2n\times 2n},
\qquad
\\*
 k_1&=\min\left\{ \tfrac{\lambda_{\min}(Q_0)}{2\|A_0\|_2+\|A_1\|_2}, \tfrac{\lambda_{\min}(Q_1)}{\|A_1\|_2}
\right\},
 \tag*{\cite[Prop.\ 1]{MelchorAguilar.2007}}
\end{align*} 
provided the equilibrium is exponentially\footnote{\label{fn:UAS_UES}
equivalently, asymptotically since 
(\ref{eq:lin_RFDE}) is a linear autonomous RFDE. In linear RFDEs with bounded delays, uniform asymptotic stability and uniform  exponential stability are equivalent \cite[Thm.~5.3 in Ch.~6]{Hale.1993}. Moreover, in autonomous or periodic RFDEs (in contrast to neutral FDEs), asymptotic stability is always uniform \cite[Lemma 1.1 in Ch.~5]{Hale.1993}. } 
stable and $Q_{0}, Q_1\succ 0_{n\times n}, Q_2\succeq 0_{n\times n}$.
Two issues should be noted. 

Firstly, since the LK functional satisfies by construction the monotonicity condition of the common LK theorem, cf.\ (\ref{eq:RFDE_k3}),   the existence of a quadratic lower bound with $k_1>0$ (or actually even $k_1\geq 0$, cf.\ Thm.~\ref{thm:nec_and_sufficient}) 
is also the crucial missing step that proves asymptotic stability 
via the LK functional. 
However, the above formulae 
are only valid if exponential (equivalently, asymptotic)   
stability has been proven beforehand. Hence, the stability analysis must already be done  by other means in a separate step. 
For instance, this can be achieved via frequency-domain based methods, e.g., via the eigenvalues of $A_y$.
Having thus $A_y$ already at hand, the approach in the present paper becomes even more convenient. 
\begin{remark}\label{rem:Psi_criterion}
As a consequence of the above issue, how at all to conclude stability from the LK functional (\ref{eq:complete_LK_with_kernels}) or the involved delay Lyapunov matrix function $\Psi$ has long been an open question. It has only recently been resolved by  Egorov et al.\ \cite{Egorov.2017} and Gomez et al.\ \cite{Gomez.2019b}. The criterion is equivalent to requiring that, for some $\tilde Q\succ 0_{n\times n}$,  
\begin{align}
\tilde P_{\mathrm{zz}}(\xi,\eta):=
\Psi(\xi-\eta;\tilde Q)
\end{align}
is a positive definite kernel, in the sense that the block matrix $(\tilde P_{\mathrm{zz}}(  \theta_j,  \theta_k))_{jk}$ 
must be positive semidefinite, 
with an a priori bound on the discretization resolution of the 
grid $(\xi,\eta) \in \{  \theta_j,  \theta_k \}_{jk}\subset[-\delay,0]\times [-\delay,0]$.  
Despite of a completely different framework, 
the result  can be brought in relation to  
  Sec.~\ref{sec:numInt} 
by rewriting the matrix in (\ref{eq:P_quad_Clenshaw_Curtis_}) as 
\begin{align}\label{eq:P_quad_Clenshaw_Curtis_structure}
P_y^{\mathrm{quad}}=S^\top (\tilde P_{\mathrm{zz}}(\tilde \theta_j,\tilde \theta_k))_{jk} S+D, 
\end{align}
with  $S=\mathrm{diag}((w_k)_k)\otimes A_1 +\left[\begin{smallmatrix} 0_{n\times nN} & I_n \\ 0_{nN\times nN} & 0_{nN\times n}\end{smallmatrix}\right]$ and $D=\mathrm{blkdiag}((w_k(Q_1+(\delay+\tilde\theta_k)Q_2))_k)$, cf.~(\ref{eq:P_num_int_complete_LK}) with (\ref{eq:P_in_completeTypeLK}).  
The first term in (\ref{eq:P_quad_Clenshaw_Curtis_structure}) clearly preserves the positive semidefiniteness of $(\tilde P_{\mathrm{zz}}(\tilde \theta_j,\tilde \theta_k))_{jk}$, and $D$ is only an added block diagonal matrix that inherits positive semidefiniteness from $Q_1,Q_2$.
\end{remark}
Secondly, of course  the LK functional changes as the delay changes. 
Note that, however, the above stated  formulae for $k_1$ do not depend on the value of the delay.

\begin{figure}
		\centering
		\captionsetup[subfigure]{justification=centering}
\subfloat[
convergence of $k_1$ (for $\delay=2$)
		\label{fig:quadratic_lower_bound_over_N_a}
	]
	{
\includegraphics{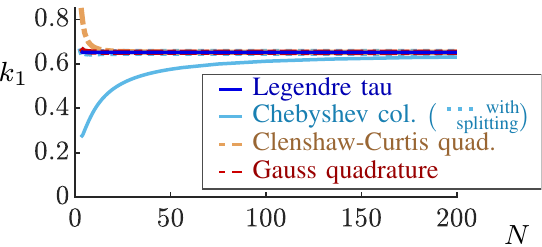}
}
\subfloat[error 
in Fig.~\ref{fig:quadratic_lower_bound_over_N_a} \label{fig:quadratic_lower_bound_over_N_b}]
	{
\includegraphics{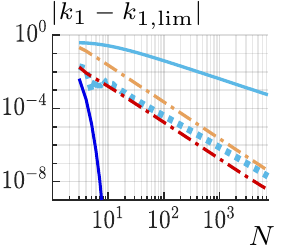}
}
\\
		\subfloat[
delay dependence of 
minimum eigenvalues in the decomposition~(\ref{eq:AitkenBlockDiag}), $P=P_y$ , $Z=P_{y,zz}$
		\label{fig:quadratic_lower_bound_over_delay_b}
	]
	{	
\includegraphics{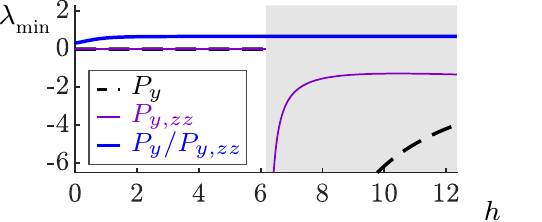}
}
		\subfloat[
error 
of the \newline critical delay  value
		\label{fig:ex2_hc_error_over_N}
	]
	{
\includegraphics{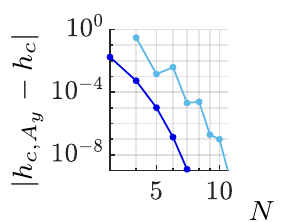}
}
\\
\subfloat[delay dependence of the bound
		\label{fig:quadratic_lower_bound_over_delay}
	]
	{

\begin{tikzpicture}[]
	\path[clip] (0,0) rectangle (5.5,3);
	
\node[inner sep=0pt,anchor=south west] (fig_pdf) at (0,0)
		    {\includegraphics[]{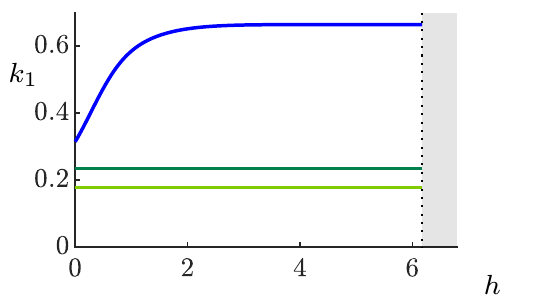}}; 
	
\node [anchor=north,text=blue,scale=0.95,] (node_Thm)
at ($(fig_pdf.north)+(0,-0.3)$)
	{\footnotesize 
	Thm.\ \ref{thm:k1}
	};

\node [anchor=north,scale=0.95, align=center,
text={rgb,1:red,0;green,0.3;blue,0.1}
] (node_1) 
at ($(fig_pdf.north)+(0,-1.3)$)
	{\footnotesize  \cite[Lem.\ 2.10]{Kharitonov.2013}  
	};
	
\node [anchor=north,scale=0.95, align=center,
text={rgb,1:red,0.2;green,0.5;blue,0}
] (node_2) 
at ($(fig_pdf.north)+(0,-1.95)$)
	{\footnotesize   \cite[Prop.\ 1]{MelchorAguilar.2007} 
	};
\end{tikzpicture}

}
		\subfloat[
{ $V(\phi)$ error}	\label{fig:Vconv}
	]
	{
\includegraphics{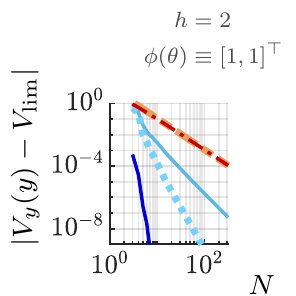}
}
	\caption{ 
	Example \ref{exmp:k1}. In particular, Fig.~\ref{fig:quadratic_lower_bound_over_delay} shows the improved quadratic lower bound. (Figures a,b,d,f share the same legend). 
	}
	\label{fig:k1}
\end{figure}
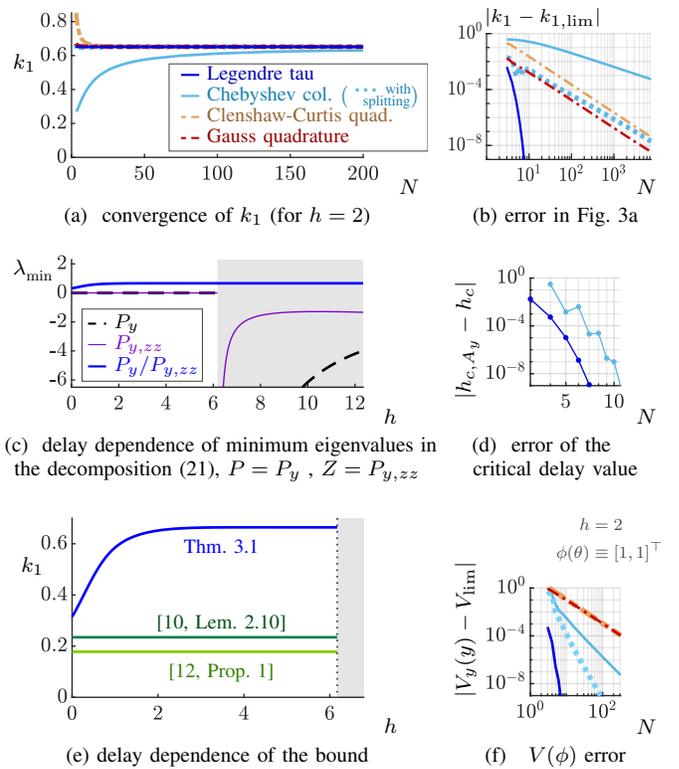

\begin{example} \label{exmp:k1}
For all delay values $\delay$ that are smaller than
$\delay_c:={\arccos(-0.9)}/{\sqrt{1-0.9^2}}
\approx 6.17$, the equilibrium of
\begin{align}\label{eq:example_RFDE}
\dot x(t)=\begin{bmatrix} -2& 0\\0& -0.9\end{bmatrix} x(t) + \begin{bmatrix} -1& 0\\-1& -1\end{bmatrix} x(t-\delay)
\end{align}
is asymptotically stable \cite[Example 3.2]{Scholl.2023}. 
Let $Q_0=Q_1=I_2, Q_2=0_{2\times 2}$.  
For any given 
$h>0$ (affecting $A_y$), 
the Lyapunov equation solution 
$P_y$ 
can be computed 
as in Example \ref{exmp:P_structure}. We get $k_1$ in (\ref{eq:kappa1_V_y}) via the 
additional lines  \vspace{0.2em}
{\small\begin{verbatim} 
 p=mat2cell(P,n*[N,1],n*[N,1]);  
 k1=min(eig(p{2,2}-p{2,1}*(p{1,1}\p{1,2})))
\end{verbatim}}\noindent 
in Matlab (as $P_{y,zz}$ is nonsingular). 
We also consider the numerical-integration-based $P_y^{\mathrm{quad}}$ from (\ref{eq:P_num_int_complete_LK}) and (\ref{eq:P_num_int_G_complete_LK}). 
Fig.~\ref{fig:quadratic_lower_bound_over_N_a} shows the convergence of $k_1$ for all approaches.
\end{example}

Fig.~\ref{fig:k1} also gives some further insights. Fig.~\ref{fig:quadratic_lower_bound_over_N_b} certificates a surprisingly fast convergence for the Legendre tau method. This is  also confirmed by other examples (if $Q_2\neq 0_{n\times n}$, the  Lyapunov equation right-hand side from Sec.~\ref{sec:splitting} below should be used).
For the Chebyshev collocation method, we are going to introduce a splitting approach 
in Sec.~\ref{sec:splitting},  
which gives  
an improved rate of convergence, cf.~Fig.~\ref{fig:quadratic_lower_bound_over_N_b}.

Fig.~\ref{fig:quadratic_lower_bound_over_delay_b} (Legendre tau, $N=1000$)  shows the interplay of the matrices in 
(\ref{eq:AitkenBlockDiag}), once the asymptotic stability is lost for delays larger than  $\delay\approx 6.17$. 
We are going to prove in Thm.~\ref{thm:nec_and_sufficient} that  positive semidefiniteness of $P_y$ is indeed necessary and sufficient for $A_y$ being Hurwitz. 

Let us consider the boundary $\delay_{c,A_y}$ between the white and gray delay region in Fig.~\ref{fig:quadratic_lower_bound_over_delay_b}. 
It marks the smallest delay  
at which the matrix $A_y$ is no longer Hurwitz (equivalently, where no longer a positive semidefinite solution $P_y$ exists), which can, e.g., be fine estimated by a bisection method.  
 Already with a rough discretization resolution $N$,  
this boundary
 reflects the analytically known critical delay $\delay_c$ of (\ref{eq:example_RFDE}) quite  precisely, and its rapid convergence is shown in Fig.~\ref{fig:ex2_hc_error_over_N} for both the Legendre tau and the Chebyshev collocation method. 

Most importantly, Fig.~\ref{fig:quadratic_lower_bound_over_delay} reveals that  the largest possible quadratic lower bound depends on the value of the delay. Thm.~\ref{thm:k1} clearly  gives a  less conservative value of $k_1$ than the   
 known formulae (green lines). 
For non-small delays, the bound is even improved by a multiple.

Fig.~\ref{fig:Vconv} shows the rapid convergence of the numerical result for  $V(\phi)$ with an exemplary argument $\phi$. 

\begin{table}
\footnotesize
\centering
\begin{tabular}{|r  c c c|}
\hline
$k_1$  &\cite[Prop.\ 1]{MelchorAguilar.2007} & \cite[Lem.\ 2.10]{Kharitonov.2013}   & Thm.\ \ref{thm:k1}
\\
\hline
\cite{MelchorAguilar.2007}, Example 5.1 & 0.7500 & 0.8229 & 1.4596
\\
\hline
\cite{MelchorAguilar.2007}, Example 5.2 & 0.6000 & 2.3238 & 3.8660
\\
\hline
\cite{MelchorAguilar.2007}, Example 5.4 & 0.1464 & 0.1978 & 0.5229
\\
\hline
\end{tabular}
\caption{Improvements of the quadratic lower bound for three physical examples from the literature. 
}
\label{tab:k1_examples}
\end{table}

A remark on non-complete functionals is in order. 
\begin{remark}\label{eq:cubic_lower_bound}
If $Q_1=Q_2=0_{n\times n}$, only a local cubic lower bound  on $V$   
is known to exist, and non-existence\footnotemark[\getrefnumber{fn:zerok1}] of a positive quadratic 
one is proven for \cite[Example~2.1]{Kharitonov.2013}. 
Indeed, for this example, $k_1$ from (\ref{eq:kappa1_V_y}) 
converges to zero as $N$ increases.
\end{remark}

 Finally, the reduced conservativity of $k_1$, already indicated by Fig.~\ref{fig:quadratic_lower_bound_over_delay}, is confirmed by other examples in Table \ref{tab:k1_examples}.

\section{Interpretation in Terms of Partial Stability}\label{sec:role_of_part_stab}
Note that $V_y$ obtained in Sec.~\ref{sec:Partial_Lyap_Eq} does not necessarily qualify as a Lyapunov function for the  ODE  (\ref{eq:y_ODE}) since, if $Q_2=0_{n\times n}$,  the matrix $Q_y$ in the Lyapunov function derivative (\ref{eq:Q_quadForm}) is not positive definite. Even if $Q_2\succ 0_{n\times n}$, 
the involved   $Q_y$ is theoretically positive definite for any finite $N$, but 
the smallest eigenvalue of $Q_y$ converges to zero as $N$ increases (the denser the grid, the smaller the integration weights $w_k$). 
Moreover, the lower bound (\ref{eq:k1_Vy}) on $V_y$ does not fit with the classical Lyapunov theory. 
The present section explains why $V_y$ is still meaningful for a  stability analysis of the approximating ODE. Within the presented approach, the lower bound  (\ref{eq:k1_Vy}) is exactly what is required. 
First, we 
clarify what we are actually looking for 
when we target stability in a RFDE.
\subsection{Stability in RFDEs}
Having in mind the classical Lyapunov theorem for ODEs, one might wonder why the lower bound in (\ref{eq:pos_def_LK}) relies on $\|x(t)\|$ and not on the norm of the RFDE state $x_t$. The latter addresses the norm in $C([-\delay,0],\mathbb R^n)$ defined by 
\begin{align}\label{eq:state_norm}
\|x_t\|_C 
 =  
\max_{\theta\in[-\delay,0]}\|x_t(\theta)\|.
\end{align}
For Lyapunov functions in ODEs, both the positive-definiteness bound (\;$\kappa_1(\|x\|)\leq V(x)$, $\kappa_1\in \mathcal K$\;) and the monotonicity requirement (\;$D_f^+V(x)\leq -\kappa_3(\|x\|)$, $\kappa_3\in \mathcal K$\;) refer to the norm of the ODE state. Thus,  one would expect 
(\ref{eq:state_norm}) at these places when transferring Lyapunov's results from $x(t)\in \mathbb R^n$ to $x_t=\phi \in C([-\delay,0],\mathbb R^n)$. 
However, this is not the case in 
the following common LK theorem -- neither in the left inequality of (\ref{eq:RFDE_k12}) nor in (\ref{eq:RFDE_k3}). Instead of $\|\phi\|_C=\|x_t\|_C$, only $\|\phi(0)\|=\|x_t(0)\|=\|x(t)\|$ occurs.  As usual, the theorem refers to  general autonomous RFDEs 
\begin{align}\label{eq:RFDE}
\dot x(t) &= f(x_t),
\end{align}
with $f(0_{n_{[-\delay,0]}})=0_n$ and  $f$ locally Lipschitz.

\begin{theorem}[LK Theorem {\cite[Thm.\ 5-2.1]{Hale.1993}}] \label{thm:LK}
If there is a continuous $V\colon C([-\delay,0],\mathbb R^n)\to \mathbb R_{\geq 0}$ such that, for all $\phi$ in a domain $
G \subseteq C([-\delay,0],\mathbb R^n)$,  $0_{n_{[-\delay,0]}}\in G$, it holds 
\begin{align}\label{eq:RFDE_k12}
\kappa_1(\|\phi(0)\|
)
\leq V(\phi)
&\leq \kappa_2(\|\phi\|_C)
\\
\label{eq:RFDE_k3}
D_{(\ref{eq:RFDE})}^+
V(\phi)&\leq -\kappa_3(\|\phi(0)\|
) ,
\end{align}
with some class-K functions $\kappa_{1,2,3}\in \mathcal K$, 
then  the zero equilibrium of (\ref{eq:RFDE}) is asymptotically stable. 
\end{theorem}

The key to the above question is that there are two, obviously equivalent, definitions of asymptotic stability in the RFDE.   
Starting from the same norm ball for the initial function $x_0$, they differ in the condition on the implication side: 
Either 
 the state 
$x_t$ 
with the norm (\ref{eq:state_norm}) is taken into account (Def.~\ref{def:RFDE_Stab}a), or the pointwise solution  
$x(t)\in \mathbb R^n$  is considered (Def.~\ref{def:RFDE_Stab}b).   
\begin{definition}[Lyapunov stability in RFDEs]
 \label{def:RFDE_Stab}
The zero equilibrium of (\ref{eq:RFDE}) is asymptotically stable if 
\\
a) $\begin{aligned} 
\begin{array}{@{}l@{\;\;\,}l@{\;\;\,}l@{\;\,}l}
\forall \varepsilon>0, 
\exists \delta
>0:
& 
\|x_0\|_C < \delta
&\Longrightarrow 
&\forall t\geq 0: \|x_t\|_C < \varepsilon
\\
\text{and} \quad \exists r>0:  &
\|x_0\|_C< r 
&\Longrightarrow  
&
\|x_t\|_C \to 0 \text{ as } t\to \infty
\end{array}
\end{aligned}$
\\
or, equivalently, \\
b) $\begin{aligned} \begin{array}{@{}l@{\;\;\,}l@{\;\;\,}l@{\;\,}l}
\forall \varepsilon>0, 
\exists \delta
>0: 
 &\|x_0\|_C < \delta
&\Longrightarrow 
&\forall t\geq 0: \|x(t)\|
 < \varepsilon
\\
\text{and} \quad \exists r>0:  
&\|x_0\|_C< r 
&\Longrightarrow 
&\|x(t)\|
 \to 0 \text{ as } t\to \infty.
\end{array}
\end{aligned}$
\end{definition}

In terms of the whole state $x_t$, the pointwise consideration in Def.\ \ref{def:RFDE_Stab}b refers only to the boundary value $x(t)=x_t(0)$ in Fig.~\ref{fig:sol_state}. 
The classical LK theorem, Thm.~\ref{thm:LK}, addresses  Def.~\ref{def:RFDE_Stab}b  since,  $\forall t\geq 0$,
\begin{align}
\kappa_1(\|x(t)\|
)\stackrel{(\ref{eq:RFDE_k12})}\leq V(x_t) \stackrel{(\ref{eq:RFDE_k3})}\leq V(x_0) \stackrel{(\ref{eq:RFDE_k12})}\leq \kappa_2(\|x_0\|_C)
\end{align}
gives a pointwise estimation $\|x(t)\|
\leq \kappa_1^{-1}(\kappa_2(\|x_0\|_C))$ to indicate  stability.
A theorem that addresses  Def.~\ref{def:RFDE_Stab}a would instead rely on $\kappa_1(\|\phi\|_C)$ in  (\ref{eq:RFDE_k12}) and 
$\kappa_3(\|\phi\|_C)$ in (\ref{eq:RFDE_k3}), as has been expected above. Such a theorem is also valid \cite[Thm.\ 30.1]{Krasovskii.1963}, but these bounds are quite restrictive and  not satisfied by common LK functionals. 
 
\subsection{Partial Stability in ODEs}
In the approximating ODE, cf.\ Fig.~\ref{fig:sol_ODE}, the state $
y(t) \in \mathbb R^{n(N+1)}$ 
represents the RFDE state $x_t\in C([-\delay,0],\mathbb R^n)$, and its last vector-valued component $y^N(t)=\hat x(t) \in \mathbb R^n$ represents the pointwise solution value $x(t)\in \mathbb R^n$.
While Def.~\ref{def:RFDE_Stab}a translates to the 
usual\footnote{The choice of the norm $\|y\|_{\infty}=\max_{k\in\{0,\ldots,N\}} \|y^k\|$ is irrelevant due to the equivalence of norms in finite-dimensional spaces.} definition of asymptotic stability in the ODE, 
Def.~\ref{def:RFDE_Stab}b amounts to the concept of partial asymptotic stability with respect to (w.r.t.) $\hat x$. 
Again, we give the definition for a general class of systems. These are ODEs where $y(t)$ is partitioned into two parts, $z(t)\in \mathbb R^p$ and $\hat x(t)\in \mathbb R^n$, with $\mathrm{dim}(y(t))=p+n$, and the latter part $\hat x(t)$ is of special interest. That is, we consider 
 autonomous~ODEs 
\begin{align} \label{eq:partial_ODE}
\frac{\mathrm d}{\mathrm d t} 
\begin{bmatrix}
z(t) \\
\hat x(t)
\end{bmatrix}
&=
\begin{bmatrix}
f^z(z(t),\hat x(t)) \\
f^x(z(t),\hat x(t))
\end{bmatrix}
\end{align}
with 
$\left[\begin{smallmatrix}
f^z(0_p,0_n) \\
f^x(0_p,0_n)
\end{smallmatrix}\right]=0_{p+n}$
and $f^{z,x}$ locally Lipschitz. 
\vspace{0.5em}
\begin{definition}[Lyapunov-Rumyantsev partial stability]\label{def:partialStab}~\newline
The zero equilibrium of (\ref{eq:partial_ODE}) is partially stable w.r.t.\ $\hat x$ if
\begin{align*}
\forall \varepsilon>0, \exists \delta>0: 
\quad 
\Big \| 
\Big[\begin{smallmatrix}
z(0) \\
\hat x(0)
\end{smallmatrix}\Big]
\Big\| <\delta \; \Longrightarrow\;  \forall t\geq 0: \|\hat x(t)\|<\varepsilon. 
\end{align*}
It is partially asymptotically stable w.r.t.\ $\hat x$ if, additionally, 
\begin{align*}
\quad\;\;
\exists r>0: \quad\;\;
\Big \| 
\Big[\begin{smallmatrix}
z(0) \\
\hat x(0)
\end{smallmatrix}\Big]
\Big\| <r \;  \Longrightarrow \; \|\hat x(t)\|\to 0 \text{ as } t\to \infty.
\\[-2.5em]
\end{align*}
\end{definition}
For an in-depth discussion of this stability concept, see \cite{Vorotnikov.1998}. As in Def.~\ref{def:RFDE_Stab}b for stability in RFDEs, the initial value deviations consider the whole state, but the implications address only the part $\hat x$ that is of special interest. 

The following partial stability  
theorem fits well with Thm.~\ref{thm:LK} (note that an upper bound $V_y(\left[\begin{smallmatrix} z \\ \hat x \end{smallmatrix} \right])\leq \kappa_2(\|\left[\begin{smallmatrix} z \\ \hat x \end{smallmatrix} \right]\|$) always exists).
\begin{theorem}[Peiffer and Rouche 1969
{\cite[Thm.\ II]{Peiffer.1969}}]. \label{thm:partialLyap}
If there is a continuous $V_y\colon\mathbb R^{p+n}\to\mathbb R_{\geq 0}$, $V_y(0_{p+n})=0$, such that, for all $\left[\begin{smallmatrix} z \\ \hat x \end{smallmatrix} \right]$ 
in a domain $
G\subseteq \mathbb R^{p+n}$, $0_{p+n}\in G$, it holds  
\begin{align}\label{eq:partial_pos_def}
\kappa_1(\|\hat x\|) \leq V_y(\left[\begin{smallmatrix} z \\ \hat x \end{smallmatrix} \right]),
\end{align}
with $\kappa_1 \in \mathcal K$, and $D_{(\ref{eq:partial_ODE})}^+V_y(\left[\begin{smallmatrix} z \\ \hat x \end{smallmatrix} \right])\leq 0$,  then the zero equilibrium of (\ref{eq:partial_ODE}) is partially stable w.r.t.\ $\hat x$. If, additionally, $\forall \left[\begin{smallmatrix} z \\ \hat x \end{smallmatrix} \right]\in G:$ 
\begin{align}\label{eq:partial_monotonicity}
 D_{(\ref{eq:partial_ODE})}^+V_y(\left[\begin{smallmatrix} z \\ \hat x \end{smallmatrix} \right]) \leq -\kappa_3(\|\hat x\|)
\end{align}
with $\kappa_3 \in \mathcal K$, and there exists $r>0$ such that $\left\|\left[\begin{smallmatrix} z(0) \\ \hat x(0) \end{smallmatrix} \right]\right\|<r$ implies that $\|f^x(z(t),\hat x(t))\|$ is bounded for all $t\geq 0$, then 
 it 
is partially asymptotically stable w.r.t.\ $\hat x$.
\end{theorem}
As in the classical LK theorem for RFDEs (Thm.~\ref{thm:LK}), both the (partial) positive-definiteness condition (\ref{eq:partial_pos_def}) and the monotonicity requirement\footnote{Criteria that come without the boundedness condition below (\ref{eq:partial_monotonicity}) require a full monotonicity condition $D_{(\ref{eq:partial_ODE})}^+V_y(y)\leq -\kappa_3(\|y\|), \kappa_3\in \mathcal K,$ cf.~\cite{Lakshmikantham.1993}.}
 (\ref{eq:partial_monotonicity}) consider only the part of special interest $\hat x(t)=y^N(t) \approx x(t)=x_t(0)$. 
We call $V$ 
in Thm.\ \ref{thm:partialLyap} a \textit{partial Lyapunov function}.

To sum up, the discretization of 
Def.~\ref{def:RFDE_Stab}b  for  RFDE stability is exactly the definition of Lyapunov-Rumyantsev partial stability w.r.t.\ $\hat x$ (Def.~\ref{def:partialStab}). 
Moreover, the Lyapunov-Krasovskii 
theorem for stability in the RFDE (Thm.~\ref{thm:LK}) becomes Peiffer and Rouche's theorem for partial stability (Thm.~\ref{thm:partialLyap}).

\subsection{Equivalence of Stability and Partial Stability in the Approximating ODE}  
In general ODEs, the concept of partial stability is a weaker concept than stability. 
We can still focus without doubt on partial stability if the equivalence between Def.\ \ref{def:RFDE_Stab}a and \ref{def:RFDE_Stab}b is reflected by the ODE approximation, so that proving partial stability w.r.t.\ $\hat x$ is already sufficient for proving stability.
\begin{condition}\label{cond:partialStabImplStab}
The zero equilibrium of the ODE approximation (\ref{eq:y_ODE}) is (asymptotically)  stable if and only if it is partially (asymptotically) stable w.r.t.\ $\hat x$. 
\end{condition}
To verify this condition for the discretization schemes at hand, we consider a 
 result from the realm of total stability.
\begin{lemma}{\cite[Thm.~3.11.3]{Lakshmikantham.1969}}. 
\label{lem:part_stab_impl_stab}
If the zero equilibrium of the auxiliary system 
\begin{align} \label{eq:partialStabImpliesStab}
\dot z &= f^z(z,0_n)
\end{align}
is asymptotically stable, then, in (\ref{eq:partial_ODE}), partial (asymptotic) stability w.r.t.\ $\hat x$ of the zero equilibrium implies (asymptotic) stability of the zero equilibrium.
\end{lemma} 
Loosely speaking, 
for 
reasonable 
approximations the latter 
seems to be a matter of course since, if $x(t)$ for $t\geq 0$ could be forced 
to remain 
zero, then, for $t\geq \delay$, the solution segment $x_t$ is zero,  
which should at least asymptotically be reflected by 
${z(t)\to 0_p}$ as $t\to\infty$. 
In 
terms of the linear ODE (\ref{eq:y_ODE}), 
Lemma~\ref{lem:part_stab_impl_stab} only refers to the  
submatrix $A_{\ind y, zz}:=({A_{\ind y}}^{jk})_{j,k\in\{0,\ldots,N-1\}}$. 
For collocation schemes like $A_{\ind y}=A_{\ind y}^C$  in Appendix \ref{sec:ChebCol}, 
 stability of this submatrix is clearly neither affected by the RFDE coefficient matrices $A_0,A_1$ (occurrence only in the last block-row), nor the delay $h$ (scalar factor), nor the dimension $n$ (Kronecker product with $I_n$). 
For tau methods, an analogous independence can be achieved by first applying a change of basis w.r.t.\ the $z$-coordinates. In appropriate coordinates, setting, e.g., $A_0=A_1=0_{n\times n}$ does not alter the submatrix eigenvalues. 
The next lemma formulates the thus motivated coordinate invariant form of Lemma \ref{lem:part_stab_impl_stab} for the linear ODE. 
 Whether it applies is, consequently, no question of $A_0,A_1,\delay$, but it is rather a question of the discretization scheme. 
\begin{corollary}\label{cor:partialStablImpl}
Consider the linear ODE (\ref{eq:y_ODE}).
If there exists a change of coordinates w.r.t.\ $z$, where   
$[z^\top, \hat x^\top]^\top = T \,[v^\top,\hat x^\top]^\top$,    
such that the left upper $nN\times nN$ submatrix 
of $T^{-1} A_y T$   
is Hurwitz, then Condition \ref{cond:partialStabImplStab} holds. 
\end{corollary}
\begin{proof}
Lemma \ref{lem:part_stab_impl_stab} with  (\ref{eq:partial_ODE}) given by $\frac{\mathrm d}{\mathrm dt}\left[\begin{smallmatrix} v\\{\hat x}\end{smallmatrix}\right]= T^{-1} A_y T \left [\begin{smallmatrix} v\\  {\hat x}\end{smallmatrix}\right]  $.  
\end{proof}
For $A_{\ind y}=A_{\ind y}^C$ from the Chebyshev collocation method (\ref{sec:ChebCol}), the submatrix $A_{\ind y, zz}$ 
can indeed proven to be  Hurwitz for any discretization resolution $N$ \cite[Prop.~2]{Diekmann.2019},\cite{Solomonoff.1989}. Thus, by 
Corollary \ref{cor:partialStablImpl} 
 (with $v=z$), Condition \ref{cond:partialStabImplStab} holds. For other discretization schemes we refer to \cite[Sec.~4.3.2]{Canuto.2006}. 
For the Legendre tau method, we consider 
the coordinates described in (\ref{eq:chi}), where $v=[(\zeta^0)^\top, \ldots, (\zeta^{N-1})^\top]^\top$ consists of the first $N$ of the $N+1$ Legendre coordinates. 
Then, Corollary~\ref{cor:partialStablImpl} (with $T^{-1} A_y T=T_{\chi\zeta }  A_\zeta T_{\chi\zeta}^{-1}$ and $T_{\chi\zeta}$ from (\ref{eq:chi}))
can numerically 
 shown to be true
 for relevant values of $N$. 

Consequently, Condition \ref{cond:partialStabImplStab} 
is not only a reasonable assumption for an ODE that approximates an RFDE, but, regarding the discretization of a RFDE, it can even be confirmed as a property of the underlying  discretization schemes.

\subsection{
Proving Stability in the ODE via $V_y(y)$}\label{sec:provingStabVy}

The main result of this section, Thm.~\ref{thm:V_as_partialLyap},  shows that 
$V_{\ind y}$ 
from Sec.~\ref{sec:Partial_Lyap_Eq}  
indeed always qualifies as a partial Lyapunov function for (\ref{eq:y_ODE}) if the equilibrium is asymptotically stable. 
As a side effect, Thm.\ \ref{thm:nec_and_sufficient} gives a necessary and sufficient stability criterion in terms of $P_y$. 
We introduce the following wording.
\begin{definition}\label{def:partPosDef}
Let $\hat x$-pd be the abbreviation for 'partially positive definite w.r.t.\ the components $\hat x$'.  
We call \\ 
(a) a function $U:  \mathbb R^{p+n}\to\mathbb R;\; y=[\begin{smallmatrix} z \\ \hat x \end{smallmatrix}]\mapsto U(y)$ 
$\hat x$-pd on 
$\Omega\subseteq \mathbb R^{p+n}$ if 
it is positive semidefinite ($\forall y \in \Omega : U(y)\geq 0$, $U(0_{p+n})=0$) 
and
$
\forall y=[\begin{smallmatrix} z \\ \hat x \end{smallmatrix}]\in \Omega \text{ with } \|\hat x\|\neq 0: 
U(y)>0.  
$ \\
(b) a symmetric matrix $M =M^\top\in \mathbb R^{(p+n)\times (p+n)}$ 
$\hat x$-pd if  $U(y)= y^\top M y$ is $\hat x$-pd on 
$\mathbb R^{p+n}$.
\end{definition}
Analogously to local, or in terms of  $U(y)=y^\top M y$ even global, positive definiteness, cf.\ \cite[Lemma 4.3]{Khalil.2002}, partial positive definiteness can be expressed via a class-K function.
\begin{lemma}\label{lem:xpd}
$M =M^\top\in \mathbb R^{(p+n)\times (p+n)}$ is $\hat x$-pd if and only if $\exists \kappa\in \mathcal K$ such that $\forall [\begin{smallmatrix} z \\ \hat x \end{smallmatrix}]\in \mathbb R^{p+n}: \kappa(\|\hat x\|)\leq [\begin{smallmatrix} z \\ \hat x \end{smallmatrix}]^\top M [\begin{smallmatrix} z \\ \hat x \end{smallmatrix}]$. 
\end{lemma}
Regarding $y^\top \!Q_y y=-D_{(\ref{eq:partial_ODE})}^+\!V_y(y)$, 
Lemma \ref{lem:xpd} refers to the class-K function  in 
(\ref{eq:partial_monotonicity}). For  $Q_y$ in  (\ref{eq:Q_quadForm}) or (\ref{eq:Qzeta2}), we can choose 
\begin{align}\label{eq:kappa3_Qy}
\kappa_3(\|\hat x\|_2):=\lambda_{\min}(Q_0
) \,\|\hat x\|_2^2\leq y^\top Q_y y.
\end{align}  
Rather decisive is whether the Lyapunov equation solution $P_y$ is also $\hat x$-pd, as it is required in (\ref{eq:partial_pos_def}) with $V_y(y)=y^\top P_y \,y$. 
\begin{lemma}\label{lem:P_xpd}
Let  $P_{\ind y}=P_y^\top$ be a 
solution of 
(\ref{eq:LyapEq_y}) for a $\hat x$-pd $Q_{\ind y}$. 
If $P_y$ is positive semidefinite, then it is even $\hat x$-pd.
\end{lemma}
\begin{proof}
The result is shown by contradiction. Assume 
there exists a $y_c=[\begin{smallmatrix} z_c \\ \hat x_c \end{smallmatrix}]$ with $\|\hat x_c\|\neq 0$ such that $y_c^\top P_{\ind y} y_c=0$. 
Then    
$P_{\ind y} y_c 
=0_{n(N+1)}$ 
(cf.\ a decomposition $P_y=C^\top C$ in  $y_c^\top P_y y_c=\|C y_c\|_2^2 =0$, $C^\top C y_c=0_{n(N+1)}$), 
which leads by (\ref{eq:LyapEq_y}) to $y_c^\top Q_{\ind y} y_c=0$, contradicting that $Q_{\ind y}$ is $\hat x$-pd. 
\end{proof}

\begin{lemma}\label{lem:V_as_partialLyap_ifPpsd} 
Let  $P_{\ind y}=P_y^\top$ be a solution of (\ref{eq:LyapEq_y}) for a $\hat x$-pd $Q_{\ind y}$. 
Consider Thm.~\ref{thm:partialLyap} in terms of 
 partial asymptotic stability  w.r.t.\ $y^N=\hat x$ for the zero equilibrium in (\ref{eq:y_ODE}).  
If $P_y$ is positive semidefinite, 
then, under Cond.~\ref{cond:partialStabImplStab}, $V_{\ind y}(y)=y^\top P_{\ind y} \,y$ satisfies the conditions on a partial Lyapunov function 
in Thm. \ref{thm:partialLyap}. 
\end{lemma}
\begin{proof}
In Thm.~\ref{thm:partialLyap},   (\ref{eq:partial_pos_def}) and (\ref{eq:partial_monotonicity}) hold by Lemma~\ref{lem:P_xpd} and
\ref{lem:xpd}. 
The  boundedness condition on $\|f^x(z(t),\hat x(t))\|$ in Thm. \ref{thm:partialLyap} 
is also ensured: 
due to Cond.~\ref{cond:partialStabImplStab}, the already provable partial stability implies stability, which is accompanied by compactness of the trajectories, and the image under the continuous mapping $f^x$ remains compact.
\end{proof}

We are led to the desired interpretation of the function $V_y$ whenever the ODE equilibrium is   asymptotically stable. 

\begin{theorem} \label{thm:V_as_partialLyap}
If $A_{\ind y}$ is Hurwitz  and Cond.~\ref{cond:partialStabImplStab} applies, then 
$V_{\ind y}$ 
from Sec.~\ref{sec:Partial_Lyap_Eq} 
is a partial Lyapunov function for 
(\ref{eq:y_ODE}).
\end{theorem}
\begin{proof}
If $A_y$ is Hurwitz, 
$P_y$ is positive semidefinite by Prop.~\ref{prop:properties}. As $Q_y$ in Sec.~\ref{sec:Partial_Lyap_Eq} is  $\hat x$-pd, Lemma~\ref{lem:V_as_partialLyap_ifPpsd} applies.  
\end{proof}

Our focus is not preliminary on a stability criterion in terms of $P_y$ because we can simply compute the eigenvalues of $A_y$ to conclude stability. 
Nevertheless, the following result might still be of interest since it shows that $V_y$ must only be tested for positive semidefiniteness. Proving existence of $\kappa_1$ in (\ref{eq:partial_pos_def}) is not required due to Lemma \ref{lem:P_xpd}.
 
\begin{theorem}\label{thm:nec_and_sufficient}
Assume  Cond.~\ref{cond:partialStabImplStab} holds. Let $P_y=P_y^\top$ be a solution of 
(\ref{eq:LyapEq_y}) for a $\hat x$-pd matrix $Q_y$  (e.g., 
(\ref{eq:Q_quadForm}) or (\ref{eq:Qzeta2})). 
The zero equilibrium of the approximating ODE (\ref{eq:y_ODE}) is 
asymp\-tot\-ically stable if and only if $P_y$ is positive semidefinite. 
\end{theorem}
\begin{proof}
If $P_y\succeq 0_{n(N+1)\times n(N+1)}$, 
then 
Lemma~\ref{lem:V_as_partialLyap_ifPpsd} applies. Thus, partial asymptotic stability w.r.t.\ $\hat x$ can be proven by Thm.~\ref{thm:partialLyap}.  
The latter implies asymptotic stability by Cond.~\ref{cond:partialStabImplStab}. Conversely, if $A_y$ is Hurwitz, then $P_y\succeq 0_{n(N+1)\times n(N+1)}$ 
because of Prop.~\ref{prop:properties}.
\end{proof}

We conclude Sec.~\ref{sec:role_of_part_stab} as follows. The function $V_y$ obtained in Sec.~\ref{sec:Partial_Lyap_Eq} does not necessarily qualify as a classical  Lyapunov function. Instead, 
it is a partial Lyapunov function for a system  
in which 
proving partial stability is already sufficient for proving stability.

\section{Convergence}\label{sec:Convergence}
A sequence of refined results with enlarged $N$ should always be  considered. 
It remains to discuss convergence aspects.

\subsection{Stability Properties of the Approximating ODEs} \label{sec:stabPreserving}
The discretization scheme used in the proposed ODE-based approach should 
be stability preserving in the following sense. 
\begin{condition}\label{cond:stabilityPreservation}
Provided the discretization resolution $N$ is chosen sufficiently large, the zero equilibrium of the approximating ODE is  exponentially stable  if and only if the zero equilibrium of the RFDE is exponentially\footnotemark[\getrefnumber{fn:UAS_UES}] stable. 
\end{condition}

Chebyshev collocation has successfully been applied 
in various fields \cite{Breda.2005,Breda.2016,Breda.2018,Wolff.2021,Michiels.2019}   where Cond.~\ref{cond:stabilityPreservation} is also  desirable. 
It is known that eigenvalues of $A_y^C$ converge to the characteristic roots of the RFDE, i.e., to the solutions $s$ of $\det(sI_n-A_0-\mathrm e^{-s\delay} A_1)=0$, or, equivalently, to the eigenvalues of the infinitesimal generator 
of the $C_0$-semigroup 
 of solution operators, see \cite{Breda.2005}. 
The red points in Fig.~\ref{fig:ew_Cheb} show 
typical 
eigenvalue chains in RFDEs, and the crosses and circles demonstrate how this chain is approached by the eigenvalues of $A_y^C$.  
There are also some additional spurious eigenvalues that do not match with RFDE characteristic roots. These, however, are easily identifiable 
as they do not persist when $N$ changes \cite[Prop.~3.7]{Breda.2005}. See,  in Fig.~\ref{fig:ew_Cheb}, the crosses ($N=40$) that do not match with circles ($N=80$). Moreover, from numerical observations, they are not expected to 
hamper Cond.~\ref{cond:stabilityPreservation}, 
see also the discussions in 
\cite[p.~361]{Wolff.2021}, \cite[p.~853]{Michiels.2019}. Thus, despite of not  being  proven, Cond.~\ref{cond:stabilityPreservation} 
in practice
is 
a tenable assumption 
for the Chebyshev Collocation method. 

The Legendre tau method is similarly powerful in approximating eigenvalues, see  Fig.~\ref{fig:ew_Leg}. 
Stability preservation of this method  (Cond.~\ref{cond:stabilityPreservation}) is proven in {\cite[Thm.~5.3]{Ito.1985}}.

As a consequence, the stability-dependent characterization of $P_y$ from 
Sec.~\ref{sec:provingStabVy} is also meaningful 
for the RFDE.  
\begin{corollary}\label{thm:N_large_nec_and_sufficient}
Assume the 
discretization scheme satisfies Cond.~\ref{cond:stabilityPreservation} and \ref{cond:partialStabImplStab}.  
Provided $N$ is sufficiently large, then $P_y$ 
from Sec.~\ref{sec:Partial_Lyap_Eq} is positive semidefinite if and only if the zero equilibrium of the RFDE is 
asymptotically stable. 
\end{corollary}
\begin{proof}
Thm.~\ref{thm:nec_and_sufficient} combined with Cond.~\ref{cond:stabilityPreservation}.
\end{proof}

\begin{figure}
		\centering
\subfloat[
Chebyshev collocation: \newline
eigenvalues of $A_y^C$ in (\ref{eq:A_y_C})
		\label{fig:ew_Cheb}
	]
	{
\includegraphics{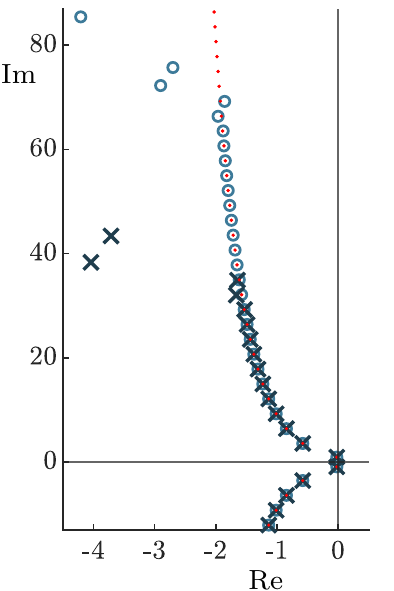}
}
\subfloat[Legendre tau:\quad eigenvalues of \newline 
 $A_\zeta^L$  in (\ref{eq:A_zeta_L}) 
or, equivalently,   
$A_y^L$  (\ref{eq:A_y_L}) \label{fig:ew_Leg}
]
	{
	\includegraphics{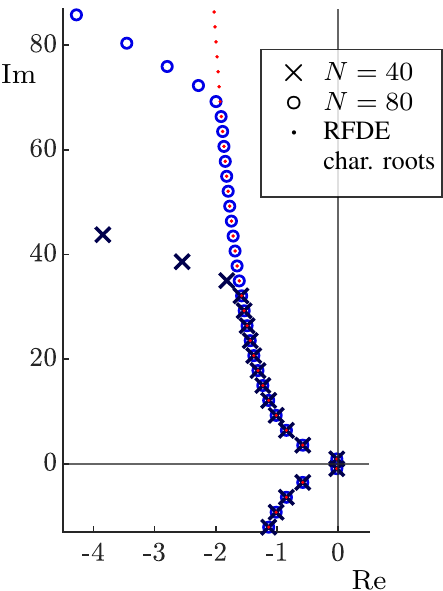}
}
	\caption{
	Characteristic roots ($A_0,A_1,\delay$ from Example \ref{exmp:P_structure}). 
	}
	\label{fig:ew}
\end{figure}
\subsection{Scheme-Dependent Improvements}\label{sec:splitting}
\subsubsection{Chebyshev collocation}  
Consider the ODE-based approach with the Chebyshev collocation method. To improve the convergence properties (indicated in Fig.~\ref{fig:k1}), we 
transform the problem of approximating $V(\phi)$ 
to a problem of approximating a modified $V_0(\phi)$  with $Q_1$ and $Q_2$ being zero. 
To this end, we choose the shift matrices $\tilde Q_1=Q_1$ and $\tilde Q_2=Q_2$ in the following splitting lemma. The idea is closely related to the derivation of complete-type functionals in  \cite[Thm.~2.11]{Kharitonov.2013}. 
\begin{lemma}[Splitting] \label{lem:splitting_V}
For $Q_0,Q_1,Q_2\in\mathbb R^{n\times n}$, let $V(\phi)=V(\phi; \;Q_0,\,Q_1,\,Q_2)$ denote a solution of (\ref{eq:DfV_complete}). 
Then
\begin{align} \label{eq:splitting_V}
&\hspace{-2.4em}V(\phi; \;\;  Q_{0},\;Q_{1},\;Q_{2}\;) =V_0(\phi)
+ V_1(\phi)+V_2(\phi)   \text{ with}
\\
V_0(\phi)&=V\big(\phi;\; \;( Q_{0}+\tilde Q_1+\delay \tilde Q_2) ,\;(Q_{1} - \tilde Q_1),\;(Q_{2}-\tilde Q_2)\;\big)
\nonumber \\
V_1(\phi) &= V(\phi; -\tilde Q_1, \tilde Q_1, 0_{n\times n})=
\int_{-\delay}^0 \phi^\top\!(\eta) \tilde Q_1  \phi(\eta)\,\mathrm d \eta
\nonumber \\
V_2(\phi)&= V(\phi; -\delay \tilde Q_2 , 0_{n\times n}, \tilde Q_2)\!= 
\!\!
\int_{-\delay}^0 
\!\!\! 
\phi^\top\!(\eta) 
(\delay+\eta)\tilde Q_2\phi(\eta) \mathrm d \eta
\nonumber
\end{align}
for arbitrarily chosen shifts $\tilde Q_1, \tilde Q_2\in \mathbb R^{n\times n}$. 
\end{lemma}
\begin{proof} 
For  $\phi(\eta)=x_t(\eta)=x(t+\eta)$, the derivatives
\begin{align*}
\iftoggle{IsTwocolumn}{&}{}
\tfrac{\mathrm{d} }{\mathrm d t} \! \smallint_{t-\delay}^t    \! \! x^\top\!(\xi)\tilde Q_1 x(\xi)\,\mathrm d \xi 
\iftoggle{IsTwocolumn}{}{&}
= 
x^\top\!(t) \tilde Q_1 x(t)
-x^\top\!(t-\delay)\tilde Q_{1} x(t-\delay) ,
\\
\iftoggle{IsTwocolumn}{&}{}
\tfrac{\mathrm d }{\mathrm d t } \! \smallint_{t-\delay}^t \!\!  x^\top \! (\xi) \big[(\delay+\xi-t)\tilde Q_2 \big] x(\xi) \,\mathrm d \xi
\iftoggle{IsTwocolumn}{
\\[-0.8em]
&  \hspace{8.5em}
}{&}
=
\delay x^\top\!(t) \tilde Q_2 x(t)
-\int_{t-\delay}^t \!  x^\top\!(\xi) \tilde Q_{2} x(\xi)\,\mathrm d \xi 
\\[-1.6em]\nonumber
\end{align*}
give $D_{(\ref{eq:lin_RFDE})}^+V_1(x_t)$ and $D_{(\ref{eq:lin_RFDE})}^+V_2(x_t)$.   They compensate in (\ref{eq:splitting_V}) the difference between $D_{(\ref{eq:lin_RFDE})}^+V_0(x_t)$ and $D_{(\ref{eq:lin_RFDE})}^+V(x_t)$ from (\ref{eq:DfV_complete}).   
\end{proof}
The first term  $V_0(\phi)$  in (\ref{eq:splitting_V}) can be approximated by $y^\top P_{y,0}\,y$  
from a Lyapunov equation  
with $Q_0$ in  (\ref{eq:Q_quadForm}) being replaced by $Q_{0}+Q_{1}+\delay Q_{2}$, and $Q_1$ and $Q_2$ being replaced by zero.
Since $V_1(\phi)$ and $V_2(\phi)$  in (\ref{eq:splitting_V})  are analytically known, 
these terms can be treated by a numerical integration. Their contributions are added on the \mbox{(block-)}diagonal
of $P_{y,0}$, i.e. $V(\phi)\approx y^\top  P_{y} \,y$, 
\begin{align*}
P_{y}=P_{y,0}
+
\mathrm{diag}((w_k)_k)\otimes Q_1
+ 
\mathrm{diag}((w_k(\delay+\tilde \theta_k))_k)\otimes Q_2,
\end{align*} 
where $w_k$ are integration weights, cf.\ Appendix \ref{sec:Clenshaw-Curtis}.

\subsubsection{Legendre tau} 
A separate numerical treatment of $V_1$ and $V_2$ in (\ref{eq:splitting_V}) is not required if the Legendre-tau-based approach is used. However, if $Q_2$ is nonzero, 
the following modification of $Q_y$ in
 (\ref{eq:Q_quadForm}) 
should be used in (\ref{eq:Lyap_eq_zeta}) or (\ref{eq:LyapEq_y})  
\begin{align}
&Q_{y}=\mathrm{blkdiag}(Q_1,0_{n(N-1)\times n(N-1)}, Q_0) + T_{\zeta y}^\top Q_{\zeta,2}  T_{\zeta y}\nonumber 
\\
&\text{with }Q_{\zeta,2}:= \mathrm{diag}( [\;(\tfrac \delay 2 \tfrac{2}{2k+1})_{k\in\{0,\ldots,N-1\}},\; \delay\;])\otimes Q_2 
\label{eq:Qzeta2}
\end{align}   
(the right lower component $\delay Q_2$ in $Q_{\zeta,2}$ is motivated by Lemma \ref{lem:V2_LegendreTau} in the appendix). 
Despite of not being treated separately in the numerical approach, the arising contributions for $V_1(\phi)$ and $V_2(\phi)$ within the approximation of $V(\phi)$ are still of interest for the proofs in the next sections. 
They can be obtained by solving 
Lyapunov equations with $Q_{0,1,2}$ being replaced by the matrices behind the semicolon in $V_1(\phi)=V(\phi;\ldots)$ and $V_2(\phi)=V(\phi;\ldots)$ from Lemma \ref{lem:splitting_V}. 
Appendix \ref{sec:V12_LegendreTau} shows that
 the resulting Legendre-tau-based approximations of $V_1(\phi)$ and $V_2(\phi)$ 
 give even the exact value for any $\phi$ that is a polynomial of order $N-1$ or less.

\subsection{Convergence Towards the Functional} 
\label{sec:convLyap}
We are interested in the following convergence statement.  
\begin{condition} \label{cond:conv}
For any given $\phi \in C([-\delay,0],\mathbb R^n) $, the scalar value $V_y(y)$ 
converges to $V(\phi)$ as $N$ increases. 
\end{condition}
More formally,  we use the notation $y=\pi_y
(\phi)$ to emphasize that the discretization $y\in \mathbb R^{n(N+1)}$ is uniquely determined from $\phi\in C$ (depending on the discretization scheme). Additionally, to keep track of the discretization resolution~$N$, a superscript $\N{N}$ is added, e.g.,  in $V_y^\N{N}(\cdot)=V_y(\cdot)$ and $\pi_y^\N{N}(\cdot)=\pi_y(\cdot)$. Thus, Cond.~\ref{cond:conv} can be rewritten as 
\begin{align}\label{eq:convV}
\forall \phi\in C: \quad  V_y^\N{N}(\pi_y^\N{N}(\phi))\to V(\phi),\quad (N\to\infty).
\end{align}
Motivated by the numerical results in Sec.~\ref{sec:Example}, we focus in this section on the Legendre tau method. Moreover, for this discretization scheme, we benefit from existing convergence proofs for the approximation of algebraic Riccati equations from the context of optimal control \cite{Gibson.1983,Ito.1987,Ito.1985}. 

\subsubsection{Operator-based description}Henceforth, we use that any argument $\phi\in C$ 
for 
$V(\phi)$ gives rise to  an element 
\begin{align}
\big[\begin{smallmatrix} \phi \\ \phi(0)\end{smallmatrix}\big] \in C\times \mathbb R^n\;\subset \;L_2\times \mathbb R^n=M_2
\end{align}
in the product  space $M_2=L_2([-\delay,0],\mathbb R^n)\times \mathbb R^n$. Note that $(M_2,\langle\cdot,\cdot \rangle_{M_2})$ is a Hilbert space with the natural inner product
\begin{align} \label{eq:M2prod}
\left\langle
\big[\begin{smallmatrix}  \phi_1\\ r_1
\end{smallmatrix}\big], 
\big[\begin{smallmatrix}\phi_2\\ r_2
\end{smallmatrix}\big]
\right\rangle_{M_2^{}}
= \int_{-\delay}^0 \phi_1^\top\!(\theta) \,\phi_2(\theta)\,\mathrm d \theta +r_1^\top r_2, 
\end{align}
$\phi_{1,2}\in L_2$, $r_{1,2}\in \mathbb R^n$. 
Similarly to the well-known $V_{\mathbb R^n}(x)=\langle x,Px\rangle_{\mathbb R^n}=x^\top\! P x$ in the finite-dimensional ODE setting for $x\in \mathbb R^n$, a complete-type LK functional can be written as 
\begin{align}\label{eq:V_M2_operator}
V(\phi)=V_{M_2}^{}(\big[\begin{smallmatrix} \phi \\ \phi(0)\end{smallmatrix}\big]) = \left\langle \big[\begin{smallmatrix} \phi \\ \phi(0)\end{smallmatrix}\big], \mathscr P \big[\begin{smallmatrix} \phi \\ \phi(0)\end{smallmatrix}\big]\right\rangle_{M_2}
\end{align}
with a self-adjoint operator $\mathscr P:M_2\to M_2$. 
Consider the splitting 
$V=V_0+V_{12}$ with $V_{12}=V_1+V_2$ from Lemma \ref{lem:splitting_V} ($\tilde Q_1=Q_1,\tilde Q_2=Q_2$). 
For the first part, which becomes  
\begin{align}\label{eq:V0_P0}
V_0(\phi)
=\langle \big[\begin{smallmatrix} \phi \\ \phi(0)\end{smallmatrix}\big], \mathscr P_0 \big[\begin{smallmatrix} \phi \\ \phi(0)\end{smallmatrix}\big]\rangle_{M_2}, 
\end{align}
the self-adjoint operator $\mathscr P_0:M_2\to M_2$
is described by suboperators on $L_2$ and $\mathbb R^n$ according to
\begin{align} \label{eq:operatorP}
\mathscr P_0\begin{bmatrix} \phi \\ r \end{bmatrix}
&=
\begin{bmatrix} \mathscr P_{zz} \phi  + \mathscr P_{zx} r \\ \mathscr P_{xz} \phi  + \mathscr P_{xx} r \end{bmatrix}
=
\begin{bmatrix}v \\w \end{bmatrix}, 
\text{ with }
\\
\begin{bmatrix}v(\theta)\\w\end{bmatrix}
&=
\begin{bmatrix}
\int_{-\delay}^0 P_\mathrm{zz}(\theta,\eta)\, \phi  (\eta)\,\mathrm d \eta + P_\mathrm{zx} (\theta) \,r
\\
\int_{-\delay}^0 P_\mathrm{xz}(\eta) \,\phi (\eta)\,\mathrm d\eta + P_\mathrm{xx}\,r
\end{bmatrix}. 
\label{eq:P_operator_kernel_func}
\end{align} 
Thus, (\ref{eq:complete_LK}) is regained by
(\ref{eq:V0_P0}), using (\ref{eq:operatorP})  with $r=\phi(0)$, 
\begin{align}\label{eq:V_operator}
\iftoggle{IsTwocolumn}{&}{}
V_0(\phi )
\iftoggle{IsTwocolumn}{}{&}
\stackrel{
(\ref{eq:M2prod})}= 
\int_{-\delay}^0 \phi  ^\top (\theta) \,v(\theta)\,\mathrm d \theta + \phi ^\top\!(0) \, w 
\\
&\stackrel{(\ref{eq:P_operator_kernel_func})\!}
=
\int_{-\delay}^0 \phi ^\top (\theta) \,
\Big(\int_{-\delay}^0 P_\mathrm{zz}(\theta,\eta)\,\phi(\eta)\,\mathrm d \eta + P_\mathrm{zx} (\theta) \,\phi(0)
\Big) \,\mathrm d \theta 
\nonumber\\
&\qquad 
+ \phi^\top\!(0) \, \Big(\int_{-\delay}^0 P_\mathrm{xz}(\eta) \,\phi(\eta)\,\mathrm d\eta + P_\mathrm{xx}\,\phi(0)\Big)  
\end{align} 
(to be more precise,  (\ref{eq:complete_LK}) with $P_\mathrm{zz,diag}(\theta)\equiv 0_{n\times n}$). 
The missing part $V_{12}=V_{1}+V_{2}$ in (\ref{eq:splitting_V}) can also be written as
\begin{align}
V_{12}(\phi)
&= 
\int_{-\delay}^0 \phi ^\top \!(\theta) \,
\big(Q_1+(\delay+\theta)Q_2\big) \,\phi(\theta)
\,\mathrm d \theta
\\*
&= \langle \big[\begin{smallmatrix} \phi \\ \phi(0)\end{smallmatrix}\big], \mathscr P_{12} \big[\begin{smallmatrix} \phi \\ \phi(0)\end{smallmatrix}\big]\rangle_{M_2}
\end{align}
based on the multiplication operator 
\begin{align} \label{eq:operatorP12}
&\mathscr P_{12} \begin{bmatrix} \phi \\ r\end{bmatrix}
=
\begin{bmatrix} \mathscr P_{zz,\mathrm{diag}} \phi  \\ 0_n \end{bmatrix}
=
\begin{bmatrix}v \\0_n \end{bmatrix}, 
\text{ with }
\\*
&v(\theta)
=
P_\mathrm{zz,diag}(\theta) \phi(\theta) = 
(Q_1+(\delay+\theta)Q_2) \phi(\theta).
\nonumber 
\end{align}
Nevertheless, we are going to treat $V_{12}$ separately\footnote{The term $\phi^\top\!(-\delay) Q_1 \phi(-\delay)$ would require an unbounded operator $\mathscr Q$ in the Lyapunov equation (\ref{eq:operator_Lyap_eq1}). Moreover, $\mathscr P_{12}$ is not compact. 
 }. 

\subsubsection{Convergence towards $V_0$} The operator $\mathscr P_0$ in (\ref{eq:V0_P0}) satisfies an operator-valued Lyapunov equation, cf.~\cite{Curtain.2020,Plischke.2005}. 
Its right-hand side is based on the right-hand side of (\ref{eq:DfV_complete}). Because of the splitting approach, the latter is $D_{(\ref{eq:lin_RFDE})}^+V_0(x_t)=x^\top(t) \tilde Q x(t)$ with $\tilde Q=Q_0+Q_1+\delay Q_2$, or, for $x_t=\phi$,  
\begin{align}
D_{(\ref{eq:lin_RFDE})}^+V_0(\phi)=-\phi^\top\!(0)\, \tilde Q\, \phi(0)=-\langle  \big[\begin{smallmatrix} \phi \\ \! \phi(0)\!\end{smallmatrix}\big], \mathscr Q  \big[\begin{smallmatrix} \phi \\ \!\phi(0)\!\end{smallmatrix}\big]\rangle_{M_2},
\end{align} 
 $ \mathscr Q \big[\begin{smallmatrix} \phi \\ \phi(0) \end{smallmatrix}\big] =  \big[\begin{smallmatrix} 0_{n_{[-\delay,0]}} \\ \tilde Q \,\phi(0) \end{smallmatrix}\big]$. 
Therefore, the operator-valued Lyapunov equation for the self-adjoint operator $\mathscr P_0=\mathscr P_0^*$ 
reads 
\begin{align} \label{eq:operator_Lyap_eq1}
\underbrace{
\langle \psi ,\mathscr P_0 \mathscr A \psi\rangle_{M_2} + \langle \psi , \mathscr A^*\mathscr P_0 \psi\rangle_{M_2}
}_{= 2 \langle  \psi , \mathscr P_0 \mathscr A \psi\rangle_{M_2} }
  = - \langle \psi ,\mathscr Q\psi\rangle_{M_2},  
\end{align}
$
\forall \psi \in D(\mathscr A)\subset M_2$, cf.~\cite{Curtain.2020,Plischke.2005}, 
  where $\mathscr A$ is the infinitesimal generator of the $C_0$-semigroup of solution operators on $M_2$ (which for linear RFDEs   is as well an appropriate state space), and $D(\mathscr A)$ is its domain. See, e.g.,  \cite{Curtain.2020} for  background on $\mathscr A$.

From the ODE-based approach in Sec.~\ref{sec:Partial_Lyap_Eq}, we  obtain  an approximation  $V_0(\phi)\approx V_{y,0}(y)=y^\top P_{y,0}\, y$, 
or, in the  
 notation of (\ref{eq:convV}), 
$V_{y,0}^\N{N} (\pi_y^\N{N} (\phi))$. 
Similarily to the exact 
$V_0(\phi)$ in (\ref{eq:V0_P0}),  
this approximation 
can be described  via 
\begin{align} \label{eq:V0N_P0N}
V_{y,0}^\N{N} (\pi_y^\N{N} (\phi)) 
= \langle \big[\begin{smallmatrix} \phi \\ \phi(0)\end{smallmatrix}\big], \mathscr P_0^\N{N} \big[\begin{smallmatrix} \phi \\ \phi(0)\end{smallmatrix}\big]\rangle_{M_2}
\end{align} 
with an approximated operator $\mathscr P_0^\N{N}$. 
Moreover, similarily to the exact operator $\mathscr P_0$ from  (\ref{eq:operator_Lyap_eq1}), this approximated operator $\mathscr P_0^\N{N}$ also satisfies an operator-valued Lyapunov equation, 
\begin{align} \label{eq:operator_Lyap_eq_approx}
2 \langle \psi ,\mathscr P_0^\N{N} \mathscr A^\N{N} \psi\rangle_{M_2} 
 = - \langle \psi ,\mathscr Q \psi\rangle_{M_2}  ,
\end{align}
which, however,  only 
relies on an approximation $\mathscr A^\N{N}$ instead of $\mathscr A$. See \cite{Ito.1987} 
for details.  The matrices  $A_\zeta$ or, equivalently, $A_y$ in Sec.~\ref{sec:ODE_Approx} are coordinate representations
of 
that 
$\mathscr A^\N{N}$.

It has to be shown that, $\forall \phi\in C$, the scalar value $V_0(\phi)$ in (\ref{eq:V0_P0})  is indeed  the limit of its approximations in (\ref{eq:V0N_P0N}) as ${N\to\infty}$. In terms of the operators, weak\footnote{\label{fn:weakConv}
The operator sequence $\{\mathscr P^\N{N}\}_{N}$ converges weakly to $\mathscr P$ if $\forall \varphi,\psi \in M_2$: $\lim\limits_{N\to\infty} \langle \varphi, \mathscr P ^\N{N}\psi\rangle_{M_2}
= \langle \varphi, \mathscr P \psi\rangle_{M_2}$  (i.e., $\forall \psi\in M_2:$ $\mathscr P^N \psi\stackrel{\text{weakly}}\to\mathscr P \psi$). 
It converges strongly if $\forall\psi \in M_2: \lim\limits_{N\to\infty} \| \mathscr P ^\N{N} \!\psi-\mathscr P \psi\|_{M_2}^{}\!=0$.
\\
The implications `operator norm conv.' $\Rightarrow$ `strong conv.' $\Rightarrow $ `weak conv.' hold. 
} 
 operator convergence $\mathscr P_0^\N{N}\stackrel{\text{weakly}}\to \mathscr P_0$ suffices for that objective. 
\begin{lemma}\label{lem:P0_weak_conv} 
Let (\ref{eq:V0N_P0N}) describe a Legendre-tau-based result 
for $V_0(\phi)$. Assume 
$\{
\|
\mathscr P_0^\N{N}
\|
\}_N$ is bounded\footnote{
To compute 
the norm of 
the self-adjoint operator $\mathscr P_0^\N{N}$  via  $P_{y,0}$, $P_{\zeta,0}$, or $P_{\chi,0}$, note that $\|\mathscr P_0^\N{N}\|=\sup_{\|(\phi,\phi(0))\|_{M_2}\leq 1}
\langle \big[\begin{smallmatrix} \phi \\ \phi(0)\end{smallmatrix}\big], \mathscr P_0^\N{N} \big[\begin{smallmatrix} \phi \\ \phi(0)\end{smallmatrix}\big]\rangle_{M_2}$. 
Considering, e.g., Sec.~\ref{sec:combined_coordinates}, we obtain $\|\mathscr P_0^\N{N}\|=\|P_{\tilde \chi,0}\|_2^{}$  where $P_{\tilde \chi,0}=T_{\tilde\chi \chi}^{-\top}P_{\chi,0} T_{\tilde\chi \chi}^{-1}$ and where $T_{\tilde\chi \chi}$ is such that $\inf_{\phi}\|(\phi,\phi(0))\|_{M_2}^2=\chi^\top G_\chi \chi=\|T_{\tilde\chi \chi} \chi\|_2^2$ with the infimum being taken over all $(\phi,\phi(0))\in M_2$  that have the discretization $\chi$. 
The latter is attained by (\ref{eq:DiscontinuousEndpoint}) with the  metric coefficients $G_\chi=\mathrm{diag}([(\frac \delay 2 \frac 2 {2k+1})_{k=0\ldots N-1},1])\otimes I_n$. Thus,   
$T_{\tilde\chi \chi}=G_\chi^{1/2}$.
}, and the existence and uniqueness conditions from Lemma \ref{lem:LyapEquationProperties} and  Rem.~\ref{rem:LyapCond} hold. Then 
$\mathscr P_0^\N{N}$ converges weakly to $\mathscr P_0$ as $N\to\infty$. 
\end{lemma}
\begin{proof} 
See \cite[Thm.~5.1 (i)]{Ito.1987}  with zero input operator and the uniqueness conditions from Sec.~\ref{sec:ExistenceUniqueness}.  
\end{proof}
In fact, this result is not at all special to the Legendre tau   method. An alternative proof from \cite[Thm.~6.7]{Gibson.1983}
applies to any discretization scheme
that satisfies 
standard conditions proving 
convergence of numerical solutions for $(x_t,x(t))$ in $M_2$. 
Lemma \ref{lem:P0_weak_conv} relies on uniform boundedness and existence assumptions. 
In the following we show that these 
can be ignored in the case of an exponentially stable RFDE equilibrium. Nevertheless, while simplifying the considerations, 
stability of the equilibrium is no necessary condition in the 
derivations. 
\begin{lemma}
If the RFDE equilibrium is exponentially stable, then the assumptions in Lemma \ref{lem:P0_weak_conv} hold. 
\end{lemma}
\begin{proof}
Let $\mathscr T(t):M_2\to M_2; \big[\begin{smallmatrix} x_0 \\ x_0(0)\end{smallmatrix} \big]\mapsto \big[\begin{smallmatrix} x_t\\x(t)\end{smallmatrix} \big]=\mathscr T(t) \big[\begin{smallmatrix} x_0\\x_0(0)\end{smallmatrix} \big]$ be the solution operator, and $\mathscr T^\N{N}(t)$ its approximation (represented by $\mathrm e^{A_y^\N{N} t}$). Due to the stability preservation property from \cite[Thm.~5.3]{Ito.1985}, $\exists M\geq 1, \beta>0,\bar N\in \mathbb N$, such that  $\forall N\geq \bar N: \|\mathscr T^\N{N}(t)\|\leq M \mathrm e ^{ -\beta t}$. Therefore, the improper integral formula  $\mathscr P^\N{N}\psi = \int_0^\infty (\mathscr T^\N{N})^*(s)\,\mathscr Q \,\mathscr T^\N{N}(s) \psi \,\mathrm ds$ is applicable, see, e.g., \cite{Gibson.1983}. Thus, with  $\|\mathscr Q \|=\|\tilde Q\|_2$, the  operators $\mathscr P^\N{N}$ are uniformly  bounded by 
$
\| \mathscr P^\N{N} \|  \leq \int_0^\infty \|\tilde Q \|_2 \|\mathscr T^\N{N}(s)\|^2 \,\mathrm d s
 \leq 
\|\tilde Q\|_2 \tfrac{M^2}{2\beta}$. Moreover, the existence and uniqueness assumptions hold by Prop.~\ref{prop:properties}.
\end{proof}
The convergence 
towards $V_0(\phi)$ does not require more than the thus established weak convergence $\mathscr P_0^\N{N} \stackrel{\text{weakly}}\to \mathscr P_0$.
However, the following stronger 
result 
will become  
helpful in Sec.~\ref{sec:lowerBoundConvergence}. 
\begin{lemma}\label{lem:P_norm_conv}
Let (\ref{eq:V0N_P0N}) describe a Legendre-tau-based result 
for $V_0(\phi)$. 
If the RFDE equilibrium is exponentially stable, 
then 
$\mathscr P_0^\N{N}$ converges in operator norm  to $\mathscr P_0$, i.e.,  it holds 
$\|\mathscr P_0^\N{N}-\mathscr P_0\|\to 0$ as $ N\to\infty $. 
\end{lemma}
\begin{proof}
See \cite[Thm.~6.9]{Gibson.1983}, where even convergence in the trace norm \cite[p.~111]{Gibson.1983} is proven. 
The result requires that not only the approximations of the solution operator $\mathscr T(t)$
converge strongly\footnotemark[\getrefnumber{fn:weakConv}], but also those of its adjoint $\mathscr T^*(t)$, which for the Legendre tau method is proven in \cite[Thm.~2.2]{Ito.1985}. 
\end{proof}

\subsubsection{Convergence towards $V$} To prove Cond.~\ref{cond:conv} on convergence towards $V=V_0+V_{12}$, it only remains to include 
$V_{12}$.
\begin{theorem}\label{prop:LegendreTau_Vconv}
If the RFDE equilibrium is exponentially stable or, more generally, if the assumptions of Lemma \ref{lem:P0_weak_conv} hold, then Cond.~\ref{cond:conv} applies for the Legendre-tau-based approach 
(with the Lyapunov equation right-hand side from (\ref{eq:Qzeta2})). 
\end{theorem}
\begin{proof}
Since $P_y$ depends linearly on $Q_y$ in the Lyapunov equation (\ref{eq:LyapEq_y}), the approximation of $V$ 
is the superposition of the 
approximations of $V_0$ and $V_{12}=V_1+V_2$ from Lemma \ref{lem:splitting_V}. 
For the first one, the convergence, $\forall  \phi\in C: \quad  V_{y,0}^\N{N}(\pi_y^\N{N}(\phi))\to V_0(\phi)$ as $N\to\infty$, is 
a consequence of the weak\footnotemark[\getrefnumber{fn:weakConv}] convergence of $\mathscr P_0^\N{N}$ proven in Lemma \ref{lem:P0_weak_conv}. Concerning the second one, the lemmata in Sec.~\ref{sec:V12_LegendreTau} show that 
$V_{12}(\phi)$ 
is approximated by $V_{y,12}^\N{N}(\pi_y^\N{N}(\phi))=V_{12}(\tilde \phi^{(N-1)})$, where $\tilde \phi^{(N-1)}$ is a Legendre series truncation of $\phi$. The convergence, $\forall  \phi\in C: \quad  V_{12}(\tilde \phi^{(N-1)}) \to V_{12}(\phi)$ as $N\to\infty$, follows from the $L_2$-convergence 
of the involved Legendre series truncation, 
$\|\phi-\tilde \phi^{(N-1)}\|_{L_2}\to 0$  as $N\to \infty$  \cite[Thm.~6.2.3]{Funaro.1992}, combined with the 
continuity\footnote{\label{fn:continuous_quadraticForm}
A quadratic form $V(x)= \langle x, \mathscr P x\rangle_X$, $\mathscr P=\mathscr P^*$, in a Hilbert space $X$ is continuous if 
$\exists k>0: \langle x, \mathscr P x \rangle_X\leq k \|x\|_X^2$, which by $\inf k=\|\mathscr P\|$ holds if $\mathscr P$ is bounded.  
Note that $V_{12}(\phi)=\langle \phi, \mathscr P_{zz,\mathrm{diag}}\phi\rangle_{L_2}\leq (\|Q_1\|_2+\delay \|Q_2\|_2)\|\phi\|_{L_2}^2$. For $V_{M_2}^{}(\psi)=\langle \psi, \mathscr P \psi\rangle_{M_2}\leq k \|\psi\|_{M_2}^2$  see \cite[p.~65]{Kharitonov.2013}.
}
 of $V_{12}$ in $L_2$.  
\end{proof}

\subsection{Quadratic Lower Bound on the Functional}\label{sec:lowerBoundConvergence}

We are going to prove that, for $N\to \infty$, the  quadratic lower bound on the approximation gives also a valid quadratic lower bound on the functional. This holds for any discretization scheme satisfying Cond.~\ref{cond:conv}.  Moreover, for the Legendre tau method, the thus obtained bound will be 
shown to be tight, 
meaning that the largest possible coefficient $k_1$ in (\ref{eq:pos_def_LK}) is obtained.  

For any discretization resolution $N$, the largest possible coefficient $k_1^\N{N}$ for the bound (\ref{eq:k1_Vy}) on the  approximation $V_y^\N{N}$ is given by (\ref{eq:kappa1_V_y}). 
Note that $k_1^\N{N}$ and, similarly, the largest possible   
coefficient $k_1=k_1^{\mathrm{opt}}$  for the bound (\ref{eq:pos_def_LK})  on the functional $V$ 
are defined by 
\begin{align}
k_1^\N{N}     
=\min\limits_{\hspace{-0.25cm}\substack{z\in \mathbb R^{nN}  \\\hat x \in \mathbb R^n\setminus\{0_n\}}\hspace{-0.25cm}} \tfrac{1}{\| \hat x \|_2^2} V_y^\N{N}(\big[\begin{smallmatrix} z \\ \hat x\end{smallmatrix}\big])
, \qquad
k_1^{\mathrm{opt}}          
 =
\inf\limits_{\hspace{-0.25cm}\substack{\phi \in C
\\ \phi(0)\neq 0_n}\hspace{-0.25cm}} \tfrac{1}{\| \phi(0) \|_2^2} V(\phi )
.
\nonumber\\[-1.25em]  \label{eq:k1_inf}
\end{align}
However, since both the functional and its approximation are quadratic, with $V(c\phi)=c^2 V(\phi)$ for any $c\in \mathbb R$ in (\ref{eq:complete_LK}) and  $V_y^\N{N}(cy)=c^2 V_y^\N{N}(y)$ in (\ref{eq:Vy}), definition (\ref{eq:k1_inf}) simplifies to  
\begin{align}
k_1^\N{N}
=
\min_{\hspace{-0.25cm}\substack{z\in \mathbb R^{nN}  \\\hat x \in \mathbb R^n\setminus\{0_n\}}\hspace{-0.25cm}}
V_y^\N{N}\big( \tfrac{1}{\| \hat x \|_2 } \big[\begin{smallmatrix} z \\ \hat x\end{smallmatrix}\big]\big) \label{eq:k1N_min_quad}
, \qquad
k_1^{\mathrm{opt}}=\inf_{\hspace{-0.15cm}\substack{\phi \in C
\\ \phi(0)\neq 0_n}\hspace{-0.15cm}} V\big( \tfrac{1}{\| \phi(0) \|_2}  \phi \big) 
.\nonumber\\[-1.25em]
\end{align}
\begin{theorem}\label{prop:k1_conv}
 If Cond.~\ref{cond:conv} holds, then $k_1=\limsup\limits
_{N\to\infty} k_1^\N{N}$ is a valid quadratic lower bound coefficient in (\ref{eq:pos_def_LK}).    
\end{theorem}
\begin{proof}
Let $\phi^{}_\delta$ give a $V(\phi^{}_\delta)$ 
that is arbitrarily close to the infimum 
 in (\ref{eq:k1N_min_quad}) according to   
\begin{alignat}{5}\label{eq:phi_delta}
\forall \delta>0, \exists   \phi^{}_\delta\in C, \| \phi^{}_\delta (0) \|_2&=1:\quad  
V(  \phi^{}_\delta )<k_1^{\mathrm{opt}}+\delta
.
\end{alignat}
The assumed convergence (\ref{eq:convV}), i.e., $\forall \phi\in C$, $\forall \varepsilon>0$, $\exists \bar N(\varepsilon,\phi)\in \mathbb N$, $\forall N\geq \bar N(\varepsilon,\phi): \vert V_y^\N{N}(\pi_y^\N{N}
(\phi))-V(\phi)\vert < \varepsilon$, shows that  
\begin{align}
\forall N\geq \bar N( \tfrac\varepsilon 2,\phi^{}_\delta):
\quad 
\vert V_y^\N{N}
(\pi_y^\N{N}
(\phi^{}_\delta))-
\underbrace{
V(\phi^{}_\delta) 
}_{\stackrel{(\ref{eq:phi_delta})}<k_1^{\mathrm{opt}}+\delta} 
\vert < \tfrac\varepsilon 2,
\label{eq:k1_delta_tile_epsilon}
\\[-3em] \nonumber
\end{align}
and thus, $\forall N\geq \bar N( \tfrac\varepsilon 2,\phi^{}_\delta):$
\begin{align}
k_1^\N{N}
\stackrel{(\ref{eq:k1N_min_quad})}=\!\!\min_{\hspace{-0.25cm}\substack{z\in \mathbb R^{nN}, \\\hat x \in \mathbb R^n\setminus\{0_n\}}\hspace{-0.25cm}}
V_y^\N{N}( \tfrac{1}{\| \hat x \|_2 } \big[\begin{smallmatrix} z \\ \hat x\end{smallmatrix}\big])
&\leq
V_y^\N{N}(\pi_y^\N{N}
(\phi^{}_\delta)) 
\label{eq:k1Nsmallerk1plusEpsilon}
\\[-1.5em]
&\!\!
\stackrel{(\ref{eq:k1_delta_tile_epsilon})}< \!
V(\phi^{}_\delta)+\tfrac{\varepsilon}{2} 
\!\stackrel{(\ref{eq:phi_delta})}<\!
k_1^{\mathrm{opt}}+\delta + \tfrac{\varepsilon}{2} 
.
\nonumber
\end{align}
Choosing $\delta=\tfrac{\varepsilon}{2}$, (\ref{eq:k1Nsmallerk1plusEpsilon}) becomes  $k_1^\N{N}< k_1^{\mathrm{opt}}+\varepsilon$. Hence, $\limsup\limits_{N\to\infty} k_1^\N{N}\leq k_1^{\mathrm{opt}}$. Any $k_1\leq k_1^{\mathrm{opt}}$ is admissible in (\ref{eq:pos_def_LK}). 
\end{proof}
For the Legendre tau method, we are going to prove that $k_1^\N{N}$ 
converges to the largest admissible coefficient $k_1^{\mathrm{opt}}$. 
The proof involves the following assumption on the 
arguments  of the minimum in (\ref{eq:k1N_min_quad}): For any $N$, we consider a vector $\big[\begin{smallmatrix} z^\N{N} \\ \hat x^\N{N}\end{smallmatrix}\big]$, with $\|\hat x^\N{N}\|_2=1$, such that $V_y^\N{N}(\big[\begin{smallmatrix} z^\N{N} \\ \hat x^\N{N}\end{smallmatrix}\big])=k_1^\N{N}$.  
By (\ref{eq:DiscontinuousEndpoint}), any $\big[\begin{smallmatrix} z^\N{N} \\ \hat x^\N{N}\end{smallmatrix}\big]$  represents a function  $\phi^\N{N}$  (we use (\ref{eq:DiscontinuousEndpoint}) since the minimizing argument is not expected to be continuous at $\theta=0$). The assumption   below is that $\phi^\N{N}$ remains uniformly bounded in $L_2$, which, however, could numerically\footnote{
The $L_2$ norm of (\ref{eq:DiscontinuousEndpoint}) can be computed from 
$\|\phi^\N{N} \|_{L_2}^2=\sum_{k=0}^{N-1}\frac h 2 \frac {2}{2k+1} \|\zeta^{ k}\|_2^2$ 
using the first $N-1$ of the $N$ subvectors in $\zeta$. These are either derived via $\zeta=T_{\zeta y} \big[\begin{smallmatrix} z\\ \hat x \end{smallmatrix}\big]$, cf.\ Rem.~\ref{rem:Matlab_AyL}, where  $ z =-P_{y,zz}^{-1}P_{y,xz}^\top \hat x$ and $\hat x =v/\|v\|_2$, see Lemma \ref{lem:part_pos_def_bound_Schur}, or are directly available if Lemma \ref{lem:part_pos_def_bound_Schur} is applied to the coordinates from (\ref{eq:chi}).
} be confirmed for all tested examples that give a nonzero $k_1$. 
\begin{theorem} \label{thm:tightBound}
Consider the Legendre tau method with (\ref{eq:Qzeta2}). 
As described above, for $\phi^\N{N}$ being related to $k_1^\N{N}$, assume that $\exists \beta>0$, $\forall N: \|\phi^\N{N}\|_{L_2}<\beta$. 
Then 
the quadratic lower bound coefficient $k_1^\N{N}$  from Cor.~\ref{cor:lowerBound} converges to 
the largest possible quadratic lower bound coefficient 
on the functional in (\ref{eq:pos_def_LK}).  
\end{theorem}
\begin{proof}
We denote  by $C_{\mathrm d}$ the set of functions $\phi:[-\delay,0]\to\mathbb R^n$ that are continuous on $[-\delay,0)$ and possibly have a jump discontinuity at the end point $\phi(0^-)\neq \phi(0)$. Note that $\phi^\N{N}\in C_{\mathrm d}$. The functional $V:C\to\mathbb R$ can straightforwardly be extended to arguments in $C_{\mathrm d}$ 
since $V(\phi)=V_{M_2}((\phi,\phi(0)))$ holds by (\ref{eq:V_M2_operator}), which, in fact, is defined for all $(\phi,\phi(0))\in L_2\times \mathbb R^n$. 
Also on this extended set of arguments, the value of interest from (\ref{eq:k1N_min_quad}) is still 
the infimum $k_1^{\mathrm{opt}}=\inf_{\hspace{-0.15cm}\substack{\phi \in C_{\mathrm d} 
\\ \phi(0)\neq 0_n}\hspace{-0.15cm}} V\big( \tfrac{1}{\| \phi(0) \|_2}  \phi \big)$ (even on $L_2\times \mathbb R^n$ 
it would be 
since $V_{M_2}$ is continuous\footnotemark[\getrefnumber{fn:continuous_quadraticForm}] in $M_2=L_2\times \mathbb R^n$ and $C$ is dense in $L_2$). 
With a slight abuse of notation we do not alter the name~$V$ for the extension on $C_{\mathrm d}$. 
By construction, the discretization $\pi_y^\N{N}(\phi^\N{N})=\big[\begin{smallmatrix} z^\N{N} \\ \hat x^\N{N}\end{smallmatrix}\big]$yields an argument of the minimum in (\ref{eq:k1N_min_quad}). 
First, we have to show that $\forall \varepsilon>0, \exists \bar N_1(\varepsilon)\in \mathbb N,$ such that 
\begin{align}\label{eq:V_diff_with_k1N}
\forall N\geq \bar N_1(\varepsilon):
\quad \vert \underbrace{V_y^\N{N}(\pi_y^\N{N}(\phi^\N{N}))}_{k_1^\N{N}} - V(\phi^\N{N})\vert < \varepsilon.    
 \nonumber
\\[-1.75em]
\end{align} 
According to the splitting approach (Lemma \ref{lem:splitting_V} with $\tilde Q_1=Q_1,\tilde Q_2=Q_2$), we decompose $V$ into three parts $V(\phi^\N{N})=  V_0(\phi^\N{N}) +V_1(\phi^\N{N})+V_2(\phi^\N{N})$ and its approximation  
correspondingly. The second and third term, $V_1(\phi^\N{N})$ and $V_2(\phi^\N{N})$, do not  contribute to the error in (\ref{eq:V_diff_with_k1N}) since $\phi^\N{N}(\theta)$ is an $(N-1)$-th order polynomial on $\theta\in[-\delay,0)$ for which the approximation is exact, according to the lemmata of 
Appendix~\ref{sec:V12_LegendreTau}. Therefore, it suffices to show uniform convergence on $\cup_N\{\phi^\N{N}\}$ for the approximations of $V_0$.  
Let  
$\psi^\N{N}=(\phi^\N{N} , \phi^\N{N}(0))\in M_2$. 
By assumption, $\|\psi^\N{N}\|_{M_2}^2=\|\phi^\N{N}\|_{L_2}^2+\|\phi^\N{N}(0)\|_2^2\leq\beta^2+1$. 
Thus, using (\ref{eq:V0_P0}) and (\ref{eq:V0N_P0N}), the error in (\ref{eq:V_diff_with_k1N}) becomes $\vert  \langle \psi^\N{N}
, \mathscr P_0^\N{N} \psi^\N{N}\rangle_{M_2} - \langle \psi^\N{N}, \mathscr P_0 \psi^\N{N}\rangle_{M_2}\vert  \leq  \| \mathscr P_0^\N{N} -\mathscr P_0\| \, (\beta^2+1)$. 
By Lemma \ref{lem:P_norm_conv}, the latter converges to zero, and thus 
(\ref{eq:V_diff_with_k1N}) holds. 
Consequently, 
$\forall N\geq \bar N_1(\varepsilon)$: 
\begin{align}
k_1^\N{N} \stackrel{(\ref{eq:V_diff_with_k1N})}>  
V(\phi^\N{N})- \varepsilon 
\geq \inf_{\hspace{-0.15cm}\substack{\phi \in C_{\mathrm d}
\\ \phi(0)\neq 0_n}\hspace{-0.15cm}} V\big( \tfrac{1}{\| \phi(0) \|_2}  \phi \big)
-\varepsilon
\stackrel{(\ref{eq:k1N_min_quad})}= k_1^{\mathrm{opt}}-\varepsilon.
\nonumber
\\[-1.25em]
\label{eq:lower_k1N}
\end{align}
With $\bar N_0(\varepsilon):=\bar N(\tfrac \varepsilon 2,\phi_{\delta=\varepsilon/2})$  from Thm.~\ref{prop:k1_conv}, we obtain 
\begin{align*}
\forall N\!\geq \max\{\bar N_0(\varepsilon),\bar N_1(\varepsilon)\}: \quad  \! k_1^{\mathrm{opt}}-\varepsilon \stackrel {(\ref{eq:lower_k1N})}< k_1^\N{N} \stackrel{(\ref{eq:k1Nsmallerk1plusEpsilon})} 
<k_1^{\mathrm{opt}}+\varepsilon, 
\end{align*}
completing the proof of $\vert k_1^\N{N}-k_1^{\mathrm{opt}}\vert \to 0$ $(N\to \infty)$. 
\end{proof}

\section{Conclusion}\label{sec:Conclusion}

The present paper shows that the counterpart of LK functionals for RFDEs are not classical Lyapunov functions for ODEs, but rather they correspond to partial Lyapunov functions, i.e., Lyapunov functions that prove partial stability. The latter are still simply obtained by solving a Lyapunov equation. Using the system matrix of an  approximating ODE, the result gives an approximation of the LK functional $V(\phi)$. 
Note that {Fig.~\ref{fig:P_structure}} yields the structure of complete-type LK functionals without any prior knowledge.  
For an appropriate ODE approximation with a sufficiently large discretization resolution $N$, the involved matrix $P_y$   is positive semidefinite if and only if the  RFDE 
equilibrium is asymptotically stable. 
A formula for a partial positive-def\-initeness bound on the functional approximation is derived. 
When it is applied to the Legendre-tau ODE-based result, a rapid convergence of the resulting lower bound coefficient is observed as $N$ increases. Its limit is shown to be  the best possible quadratic lower bound coefficient $k_1$ on the LK functional. Examples demonstrate that the latter significantly improves known results. In particular, the obtained $k_1$ depends on the delay, which is not the case in  existing formulae.
For the sake of validation, the present paper also proposes a numerical integration of the 
LK functional formula by Clenshaw-Curtis and Gauss quadrature rules. For these, the lower bound formula is purposeful as well. However, the ODE-based approach is expected to provide approximations of LK functionals even in more general cases where the LK functional is not known analytically.

%% file: text_appendix.tex
Table \ref{tab:polynomialApprox} classifies the employed polynomial methods.
The following appendix also includes some implementation hints.
\renewcommand{\arraystretch}{1.5}
\begin{table*}
\centering
\iftoggle{IsTwocolumn}{\scriptsize}{\newscriptsize}
{
\begin{tabular}{|p{4cm} @{\hspace{1.25em}} @{} p{2.5cm} ||p{2.5cm}@{\hspace{1.5em}} @{}p{0.95cm} p{0.95cm}|p{2cm} @{\hspace{1.5em}} @{} p{0.95cm} p{0.95cm}|}
 \hline
\multicolumn{2}{|c||}{\textbf{polynomial approximation} (of functions)}
&
\multicolumn{3}{c|}{\textbf{spectral methods} (for differential equations)}
&
\multicolumn{3}{@{}c@{}|}{\textbf{numerical integration} (for integral expressions)}
\\
\hline
\multicolumn{2}{|p{7cm}@{}||}{
{ (*) JP stands for Chebyshev,  Legendre, or other Jacobi polynomials}
}
&
\hfill JP:
& Chebyshev & Legendre &
\hfill JP:& Chebyshev & Legendre
\\
\hline
\vspace{0.2em}\multirow{2}{4cm}{\textbf{interpolation} / 
coincidence in  the chosen nodes (natural basis: Lagrange polynomials w.r.t.\ the chosen nodes),
\newline
\vspace{0pt} \hspace{0.5em} equivalently,  \hspace{0.5em} 
`discrete expansion' with the $0$-th to $N$-th JP as basis, 
related to an approximation of the series truncation below via quadrature 
}
&
\vspace{-0.2em}Gauss-Lobatto nodes \newline (extrema of the $N$-th JP)\vspace{0.8em}
&
\multirow{2}{1.1\linewidth}{\vspace{0.1em}\newline
\textbf{collocation} / 
\newline pseudospectral method / 
\newline method of selected points /
\newline at the nodes vanishing residual
}
&
\vspace{-0.2em}
\method{Chebyshev col.}
&
&
\vspace{-0.25em}
\multirow{2}{1.1\linewidth}{
\vspace{0em}\newline
\textbf{interpolatory quadrature} /
\newline 
integration of an interpolating polynomial  instead of the original function
}
&
\vspace{-0.2em}\method{Clenshaw-Curtis}
&
\\
\cline{2-2}\cline{4-5}\cline{7-8}
&
\vspace{-0.2em}Gauss nodes  (roots of the $(N+1)$-th JP)\vspace{0.8em}
&
&\multicolumn{2}{@{}c @{}|}{\begin{tabular}{@{}c @{}}-- \hspace{4em} --\\[-0.5em]
{
\!\!(no boundary nodes for}\\[-0.75em]
{
boundary conditions)}\\[-2em]
\end{tabular}}
&&
&
\vspace{-0.2em}\method{Gauss quad.}
\\
\hline
\multicolumn{2}{|p{6.5cm}||}{
\vspace{-0.2em}
\textbf{series truncation} / 
\newline
orthogonal\footnotemark projection to the $0$-th to $N$-th JP / \newline 
`continuous expansion' with the $0$-th to $N$-th JP as basis / 
\newline
generalized Fourier truncation / 
least squares best approximation
}
&
\textbf{Galerkin-like methods}, \newline e.g., 
Galerkin method \newline or  tau method / Lanczos' tau method / Galerkin with boundary bordering
&
&
\vspace{1em}\method{Legendre tau}
&
\vspace{0.4em}
\textbf{--} 
(the projection requires itself integral evaluations)
& \vspace{2em}\hspace{1.5em}--
& \vspace{2em}\hspace{1.5em}--
\\
\hline
\end{tabular}
}
\caption{
Classification of the used  methods 
 ( ``/'' marks synonymous terms). See   
\cite{Hesthaven.2007}, \cite{Trefethen.2008,
Trefethen.2000,Canuto.2006,Funaro.1992} 
for details.
\label{tab:polynomialApprox}}
\end{table*}
\renewcommand{\arraystretch}{1}
\footnotetext{orthogonal w.r.t.\ the (weighted) inner product in which the chosen basis polynomials are orthogonal. 
In (\ref{eq:zeta_discretization}), the modified $\zeta^N$  makes the projection non-orthogonal, unless the discretization is interpreted in terms of (\ref{eq:DiscontinuousEndpoint}). 
}

\subsection{ODEs that Approximate RFDEs}

We consider ODE approximations for (\ref{eq:lin_RFDE}) from two spectral methods: Chebyshev collocation 
and Legendre tau.

\subsubsection{Chebyshev collocation method}\label{sec:ChebCol}

By interpolation,  the 
vector $y(t)$ 
at time $t$  
in (\ref{eq:y_vec}), cf.\ Fig.~\ref{fig:sol},  determines an $N$-th order approximating polynomial for $x_t$. More specifically, 
\begin{align}\label{eq:x_t_Lagrange_basis}
x_t(\theta) \approx \sum_{k=0}^{N} \,y^k(t) \;\ell_k\big(\vartheta(\theta)\big),
\end{align}
where $ \ell_k\colon[-1,1]\to \mathbb R$ are interpolating  Lagrange basis polynomials w.r.t.\ the (Gauss-Lobatto) Chebyshev nodes $\{\tilde \vartheta_k\}_{k\in\{0,\ldots,N\}}$ on $[-1,1]$, and where 
\begin{align} \label{eq:vartheta_mapping}
\vartheta\colon [-\delay,0] \to [-1,1];\qquad  \theta \mapsto \vartheta(\theta):=\tfrac{2}{\delay} \theta +1 
\end{align}
maps the argument $\theta\in [-\delay,0]$ to this interval.

The exact 
evolution of $x_t$ in Fig.\ \ref{fig:sol_state} can be described by an abstract ODE in $C([-\delay,0],\mathbb R^n)$, see \cite{Breda.2015} for details. This abstract ODE can be discretized via the collocation method. The result describes the dynamics of the unknown coefficients $y^k(t)$ in (\ref{eq:x_t_Lagrange_basis}). It is the ODE (\ref{eq:y_ODE}) with $A_ {\ind y}=A_{\ind y}^C$, 
\begin{align}\label{eq:A_y_C}
A_{\ind y}^C \!:=\!
\left[\begin{smallmatrix}
\tfrac{2}{\delay} \ell_0'(\tilde \vartheta_0) I_n 
& 
\cdots
&
\cdots
&
\cdots
& 
\tfrac{2}{\delay} \ell_N'(\tilde \vartheta_0) I_n 
\\
\vdots&&&&\vdots
\\
\tfrac{2}{\delay} \ell_0'(\tilde \vartheta_{N-1}) I_n 
& 
\cdots
&
\cdots
&
\cdots
& 
\tfrac{2}{\delay} \ell_N'(\tilde \vartheta_{N-1}) I_n 
\\
A_1 & 0_{n\times n} &\quad \cdots \quad & 0_{n\times n} & A_0
\end{smallmatrix}\right]
\!,
\end{align}
cf.\ \cite{Breda.2016}.
The upper part of $A_{\ind y}^C$ that is 
given by
$\frac 2 \delay (\ell_k'(\tilde \vartheta_j))_{j\in\{0,\ldots,N-1\},k\in\{0,\ldots,N\}} \otimes I_n$   
requires 
the first $N$ rows of the  
$(N\!\hspace{-1pt} +\! 1)\!\times \!(N\! \hspace{-1pt}+\! 1)$
differentiation matrix $(\ell_k'(\tilde \vartheta_j))_{j,k\in\{0,\ldots,N\}}$. See  \cite[p.~54]{Trefethen.2000} (with $x_k=-\tilde \vartheta_k$).   
\begin{remark}[Implementation of $A_y^C$]\label{rem:Matlab_AyC}
A Matlab implementation of 
the skew-centrosymmetric differentiation matrix 
is available from \texttt{diffmat} in the Chebfun toolbox \cite{Driscoll.2014}. Based on the latter, 
$A_y^C=\texttt{A}$ is obtained from \vspace{0.2em}  
{\small \begin{verbatim}
 D=diffmat(N+1,[-delay,0]);   A=kron(D,eye(n)); 
 A(end-n+1:end,:)=[A1,zeros(n,n*(N-1)),A0]
\end{verbatim}} \noindent
(if $A_0,A_1,\delay,n,N$ are assigned to \texttt{A0,A1,delay,n,N}).
\end{remark}
\subsubsection{Legendre tau method} \label{sec:LegendreTau}
Let $\Leg_k:[-1,1]\to\mathbb R$ denote the $k$-th Legendre polynomial. See, e.g., \cite{Hesthaven.2007} for formulae and plots.
 Using $\{p_k(\vartheta(\cdot))\}_{k=0}^N$ as basis, an $N$-th order approximating polynomial 
for $x_t$ 
becomes 
\begin{align} \label{eq:polynomial_Legendre}
x_t(\theta) \approx \sum_{k=0}^N \,\zeta^k(t) \;\Leg_k(\vartheta(\theta)).    
\end{align}
The evolution of the coefficients $\zeta^k(t)\in \mathbb R^n$, stacked as $\zeta:=[(\zeta^0)^\top,\ldots, (\zeta^N)^\top]^\top$,  
shall again be described by  
\begin{align}\label{eq:zetaODE}
\dot \zeta (t)
&= A_\zeta^{} \,
 \zeta(t).
\end{align}  
In \cite{Ito.1986}, this is achieved 
 via Lanczos' tau method (considering a Hilbert space setting, cf.\ Sec.~\ref{sec:convLyap}). A general introduction to the tau method is given in \cite{Hesthaven.2007}.  We only state the result, which is (\ref{eq:zetaODE}) with
$A_\zeta=A_\zeta^{L}$ having the block entries 
\begin{align}
A_{\zeta}^{L,jk}\!=\!
\left\{
\begin{array}{@{}l @{\;} l}
\tfrac{2}{\delay} (2j+1) I_n, &  
{\small 
\begin{array}{@{}l@{}l}
\text{if }& 
j\in\{0,\ldots,N-1\}, \\
&
 k \!>\!j \text{ and } j\!+\!k \text{ odd}
\end{array}
}
\\
A_0+(-1)^k A_1 -\tfrac{2}{\delay}  \tfrac{k(k+1)}{2} I_n,
& \text{if } j=N 
\\
0_{n\times n} ,  &\text{else}.
\end{array}
\right.
\nonumber
\\[-1.5em] \label{eq:A_zeta_L}
\end{align}
Thus, $A_\zeta^{L}$ exhibits the structure (exemplarily for $N$ even) 
\newlength{\hs}
\setlength{\hs}{2em}
\begin{align*}
A_\zeta^L&=\left[\begin{smallmatrix}
0_{n\times n} & 0_{n\times n} &0_{n\times n} 
&\cdots 
&0_{n\times n} 
\\
\vdots & \vdots & \vdots & &\vdots &
\\
0_{n\times n} & 0_{n\times n} &0_{n\times n} 
&\cdots 
&0_{n\times n}
\\[0.25em]
A_0+A_1 & A_0-A_1 & A_0+A_1
&\cdots
& A_0-A_1 
\\[0.25em]
\end{smallmatrix}\right]
\nonumber
\\[0.5em]
&\quad+\tfrac{2}{\delay}
\left[\begin{smallmatrix}
\hspace{\hs} &  \hspace{\hs} &  \hspace{\hs} &  \hspace{\hs} &  \hspace{\hs} &  \hspace{\hs}
\\
0 & 1&0 & 1& 
\cdots 
 &  0\\[0.25em]
0 &0 & 3 & 0 & 
\cdots
& 3  \\[0.25em]
0 & 0 & 0 & 5   & 
\cdots 
 &  0 \\
\vdots &&&
 \ddots 
&\ddots
& \vdots\\
0 & 0 & 0 & 0 & 
\cdots 
& (2N-1)   \\
0 & -1& -3 & -6  & 
\cdots
& -\tfrac{N(N+1)}{2} 
\end{smallmatrix}\right]
\otimes I_n.
\end{align*}
\begin{remark}[Implementation of $A_\zeta^L$]\label{rem:Matlab_AcL}
Written in standard Matlab code, 
we obtain 
$ A_\zeta^L=\texttt{Ac}$ from \vspace{0.2em} 
{\small \begin{verbatim}
 Dc=zeros(N+1,N+1);Dc(end,:)=-(0:N).*(1:N+1)/2; 
 for j=0:N-1; Dc(j+1,(j+1)+1:2:end)=2*j+1; end
 Ac=2/delay*kron(Dc,eye(n));
 Ac(end-n+1:end,:)= Ac(end-n+1:end,:)+...
      kron(ones(1,N+1),A0)+kron((-1).^(0:N),A1)
\end{verbatim}}\noindent 
(with \texttt{A0,A1,delay,n,N} as above). 
\end{remark}
Hence, we have the dynamics (\ref{eq:zetaODE}) of the   Legendre coordinates $\zeta(t)\in\mathbb R^{n(N+1)}$ that, in (\ref{eq:polynomial_Legendre}), describe  the approximating polynomial for $x_t$. 
However, we can equivalently express the 
polynomial (\ref{eq:polynomial_Legendre}) in interpolation coordinates $y(t)\in\mathbb R^{n(N+1)}$, referring to the interpolation basis $\{\ell_k(\vartheta(\cdot))\}_{k=0}^N$ of (\ref{eq:x_t_Lagrange_basis}).
Let $T_{\ind{y}\zeta}^{}$ denote the transformation matrix of this change of basis 
\begin{align}\label{eq:T_yc}
y(t)
= 
T_{\ind{y}\zeta}^{}\,
\zeta(t). 
\end{align}
The thus computed $y(t)$ 
obeys the ODE (\ref{eq:y_ODE}) with $A_{\ind y}=A_{\ind y}^L$,
\begin{align}\label{eq:A_y_L}
A_{\ind y}^L:=T_{\ind y\zeta} \,A_\zeta^L \,T_{\ind y\zeta}^{-1}.  
\end{align}
Note that the first and last 
block row 
of $T_{\ind y\zeta}$ in (\ref{eq:T_yc}) are simply
\begin{align}\label{eq:first_last_row_Tyc}
\begin{bmatrix}
y^0 (t)
\\ 
y^N (t)
\end{bmatrix}
= 
\begin{bmatrix}
I_n & -I_n  & \cdots & (-1)^N I_n \\
I_n & I_n & \cdots & I_n
\end{bmatrix}
\zeta(t)
\end{align}
since  $\Leg_k(-1)=(-1)^k$, $\Leg_k(1)=1$  (and 
$T_{\ind y\zeta}^{jk}=p_k(\vartheta(\tilde \theta_j)) I_n$).
\begin{remark}[Implementation of $A_y^L$] \label{rem:Matlab_AyL}
Efficient conversion algorithms \cite{Townsend.2018} are found in the Chebfun toolbox. 
Applying these to the identity matrix yields $T_{\ind y\zeta}$ and $T_{\zeta\ind y}:=T_{\ind y\zeta}^{-1}$. Thus, $A_y^L=\texttt{A}$ can be derived by adding the lines
{\small \begin{verbatim}
 Tyc=kron(legcoeffs2chebvals(eye(N+1)),eye(n));
 Tcy=kron(chebvals2legcoeffs(eye(N+1)),eye(n));
 A=Tyc*Ac*Tcy
 \end{verbatim}}\vspace{-1em}\noindent to the code given in Remark \ref{rem:Matlab_AcL}. 
\end{remark}
\subsubsection{Further notes on the Legendre tau method} 
\paragraph{Lyapunov equation in Legendre coordinates}
\label{sec:LegCoordinates}
To obtain an approximation of $V(\phi)$ via the Legendre tau approach, we can resort to (\ref{eq:LyapEq_y})  with $A_{\ind y}=A_{\ind y}^L$ from (\ref{eq:A_y_L}). However, from a numerical point of view, it might be  preferable to remain in Legendre coordinates $\zeta$ and to use $A_\zeta=A_\zeta^L$, (\ref{eq:A_zeta_L}), in
\begin{align} \label{eq:def_V_zeta}
V_\zeta(\zeta)&:=\zeta^\top P_\zeta \zeta,\quad P_\zeta=P_\zeta^\top\in \mathbb R^{n(N+1)\times n(N+1)} \\
D_f^+V_\zeta(\zeta) &= \zeta^\top(P_\zeta A_\zeta + A_\zeta^\top P_\zeta) \zeta \stackrel != -\zeta ^\top (T^\top_{\ind y\zeta} \,Q_{\ind y }\,T_{\ind y\zeta} )\zeta, \nonumber
\end{align}
$\forall \zeta\in \mathbb R^{n(N+1)}$ with $T_{\ind y\zeta}$ from (\ref{eq:T_yc}). That is, we solve 
\begin{align} \label{eq:Lyap_eq_zeta}
P_\zeta A_\zeta + A_\zeta^\top P_\zeta = -T^\top_{\ind y\zeta} \,Q_{\ind y} \,T_{\ind y\zeta} 
\end{align}
for $P_\zeta$ and, if desired,  express the result in $y$ coordinates
\begin{align} \label{eq:V_y_from_V_zeta}
V_{\ind y}(y)=V_\zeta(\zeta)=V_\zeta(T^{-1}_{\ind y\zeta} y) =y^\top \big(\underbrace{(T^{-1}_{\ind y\zeta})^\top P_\zeta \,T^{-1}_{\ind y\zeta}}_{=:P_{\ind y}} \big)  y.
\\[-2em]\nonumber
\end{align}
With $Q_y$ from  (\ref{eq:Qzeta2}), only the first and last block rows (\ref{eq:first_last_row_Tyc}) of $T_{y\zeta}$ are required in (\ref{eq:Lyap_eq_zeta}).

\paragraph{Discretization}
The discretization $\zeta=(\zeta^k)_{k\in\{0,\ldots,N\}}$ of a given function $\phi$, e.g., an initial condition $x_0=\phi$ or an argument of the functional $V(\phi)$, 
is chosen as   
\begin{align} \label{eq:zeta_discretization}
\zeta^k=\tilde \zeta ^k,\; \text{ if } k<N,\quad   \text{ and } \quad  \zeta^N=\phi(0)-\sum_{k=0}^{N-1}\zeta^k,
\end{align}
where $\{\zeta ^0,\ldots, \zeta^{N-1}\}$ stem from a truncation of the Legendre series representation $\phi(\theta)=\sum_{k=0}^\infty \tilde \zeta ^k p_k(\vartheta(\theta))$, \cite{Ito.1986}. The last component $\zeta^N$ in (\ref{eq:zeta_discretization}) is such that the $N$-th order approximating polynomial $ \sum_{k=0}^{N} \zeta^k p_k(\vartheta(\theta))\approx \phi(\theta)$, at $\theta=0$,    exactly matches $\phi(0)$ (note that $\vartheta(0)=1$ and $p_k(1)=1$, $\forall k$).

\paragraph{$V_1$ and $V_2$}\label{sec:V12_LegendreTau} For two important cases of the right-hand side $-Q_\zeta=-T^\top_{\ind y\zeta} \,Q_{\ind y }\,T_{\ind y\zeta}$ in (\ref{eq:Lyap_eq_zeta}), we can  give the solution $P_\zeta$, respectively the resulting 
$V_y(y)=V_\zeta(\zeta)\approx V(\phi)$, analytically. 
\begin{lemma}\label{lem:V1_LegendreTau}
The Legendre-tau-based approximation of $V_1(\phi)$ in Lemma \ref{lem:splitting_V} becomes 
$V_1(\tilde \phi^{(N-1)})$, where $\tilde \phi^{(N-1)}(\theta)=\sum_{k=0}^{N-1} \tilde \zeta^k p_k(\vartheta(\theta))$ is the $(N-1)$-th order Legendre series truncation of  $\phi(\theta)=\sum_{k=0}^{\infty} \tilde \zeta^k p_k(\vartheta(\theta))$.
\end{lemma}
\begin{proof}
For $Q_y\!=\!\mathrm{diag}([1,0_{1\times n(N-1)},-1] )\otimes \tilde Q_1$, and, thus, 
$Q_\zeta^{jk}=(-1+(-1)^{j+k})\tilde Q_1$, it can be verified that 
\begin{align}
P_\zeta=\mathrm{diag}([(\tfrac \delay 2 \tfrac{2}{ 2k+1})_{k\in\{0,\ldots,N-1\}},0])\otimes \tilde Q_1
\end{align} is a solution of (\ref{eq:Lyap_eq_zeta}). Hence, $
V_\zeta(\zeta)=
\zeta^\top P_{\zeta}\zeta=\sum_{k=0}^{N-1} \tfrac \delay 2 \frac{2}{ 2k+1} \zeta_k^\top \tilde  Q_1 \zeta_k$. Equivalence with $V_1(\tilde \phi^{(N+1)})=  \int_{-\delay}^0 (\sum_{k=0}^{N-1} \tilde \zeta^j p_j(\vartheta(\theta)) )^\top\tilde Q_1 ( \sum_{j=0}^{N-1} \tilde \zeta^k  p_k(\vartheta(\theta)))\,\mathrm d \theta$
follows from (\ref{eq:zeta_discretization})  and  $\int_{-1}^1 p_j(  \vartheta) p_k(   \vartheta)\,\mathrm d    \vartheta=  \frac{2}{ 2k+1}\delta_{jk}$ 
 \cite[B.1]{Hesthaven.2007}.
\end{proof}
\begin{lemma} \label{lem:V2_LegendreTau}
Provided (\ref{eq:Qzeta2}) is used, the Legendre-tau-based approximation of $V_2(\phi)$ in Lemma \ref{lem:splitting_V} becomes 
$V_2(\tilde \phi^{(N-1)})$ with $\tilde \phi^{(N-1)}$ as above (Lemma \ref{lem:V1_LegendreTau}). 
\end{lemma}
\begin{proof}
Consider $Q_\zeta=Q_{\zeta,0}+Q_{\zeta,2}$ with $Q_{\zeta,0}=T_{y\zeta}^\top (\mathrm{diag}([0_{1\times nN},1] )\otimes Q_0 ) T_{y\zeta} = 1_{(N+1)\times (N+1)}\otimes Q_0$ and $Q_{\zeta,2}= \mathrm{diag}([(\tfrac{1}{2k+1})_{k\in\{0,\ldots,N-1\}},1] )\otimes  \delay Q_2$, where $Q_0=-\delay \tilde Q_2, Q_2=\tilde Q_2$. 
It can be verified that $P_\zeta$ with 
\begin{align}
P_\zeta^{jk}=
\left\{ \begin{array}{l@{}l}
(\tfrac \delay 2 )^2 \frac 2 {2j+1} \frac k {2k+1} \tilde Q_2  \;\;\;& \text{ if } j=k-1<N-1  
\\
(\tfrac \delay 2 )^2 \frac 2 {2j+1} \tilde Q_2 &\text{ if }j=k<N 
\\
(\tfrac \delay 2 )^2 \frac 2 {2j+1}\frac {k+1} {2k+1} \tilde Q_2 & \text{ if } j=k+1<N 
\\
0_{n\times n} &   \text{ else } 
\end{array}\right.
\end{align}
solves (\ref{eq:Lyap_eq_zeta}). 
The equality $\zeta^\top P_\zeta \zeta=V_2(\sum_{k=0}^{N-1} \zeta^k p_k(\vartheta(\theta)))$  is shown by using the  three-term recurrence relation \cite[(4.17)]{Hesthaven.2007} $\vartheta p_k(\vartheta)=\tfrac k {2k+1} p_{k-1}(\vartheta) + \tfrac{k+1}{2k+1} p_{k+1}(\vartheta)$  in $V_2$.
\end{proof}

\paragraph{Combined coordinates} \label{sec:combined_coordinates}
Besides of Legendre coordinates~$\zeta$, and interpolation coordinates $y=\big[\begin{smallmatrix} z \\ \hat x \end{smallmatrix}\big]=T_{y\zeta} \zeta$, the combination of the first $(N-1)$ of the $N$ Legendre coordinates and the boundary value $\hat x=\phi(0)$, 
 \begin{align}
\chi&= 
\left[\begin{smallmatrix}
 \zeta^0 \\
\vdots 
\\ \zeta^{N-1}\\
 \hat x
\end{smallmatrix}\right]
\label{eq:chi}
=
\underbrace{
\begin{bmatrix}
I_{nN}^{} & 0_{n\!N\!\times \!n}^{} \\ I_n \quad \cdots\quad  I_n & I_n
\end{bmatrix}}
_{T_{\chi \zeta}^{}}
\zeta,
\end{align} 
is, in light of (\ref{eq:zeta_discretization}), 
as well an appropriate choice  of coordinates. 
In particular, Lemma~\ref{lem:part_pos_def_bound_Schur} can directly be applied to $P_\chi= T_{\zeta\chi}^\top P_\zeta T_{\zeta\chi}$, with $T_{\zeta\chi}=T_{\chi\zeta}^{-1}=\left[\begin{smallmatrix}
I_{nN} & \!\!\!0_{nN\times n}\\ 
-1_{\!N}^\top\otimes I_n \!\!\! & I_n
\end{smallmatrix}\right]$, thus obtaining the quadratic lower bound as in Cor.~\ref{cor:lowerBound}, but without the need of $P_y=T_{\zeta y}^\top P_\zeta T_{\zeta y}$ (which would require $T_{\zeta y}$ from Rem.~\ref{rem:Matlab_AyL}).

\paragraph{Discontinuous basis}
The given coordinates $\zeta$, or equivalently $y=T_{y\zeta} \zeta$, or $\chi=T_{\chi\zeta}\zeta$, uniquely represent an 
 $N$-th order approximating polynomial 
$
\phi(\theta)\approx \sum_{k=0}^N \zeta^k p_k(\vartheta(\theta))
=\sum_{k=0}^N y^k\ell_k(\vartheta(\theta))
= \sum_{k=0}^{N-1} \chi^k (p_k(\vartheta(\theta))-p_N(\vartheta(\theta)))+  \chi^N  p_N(\vartheta(\theta))
$. 
However, if a function with a jump discontinuity at $\theta=0$ 
is of interest, it is convenient\footnote{See, in \cite{Ito.1987}, the projector $Q^N$ versus $L^N$.} to consider as approximating function instead the piecewise defined \mbox{$(N-1)$-th} order polynomial with a discontinuous end point  
\begin{align}\label{eq:DiscontinuousEndpoint}
\phi(\theta)\approx \left\{
\begin{array}{ll}
\sum_{k=0}^{N-1}\zeta^k p_k(\vartheta(\theta)), &\text{ if } \theta<0
\\
\hat x, &\text{ if }\theta=0
\end{array}\right.
\end{align}  
(which, in (\ref{eq:zeta_discretization}), has the same discretization).

\subsection{
Numerical Integration of LK Functionals}\label{sec:num_int}
Sec.~\ref{sec:numInt} proposes to apply  interpolatory quadrature rules to the LK functional. We consider Clenshaw-Curtis and Gauss quadrature. See, e.g., \cite{Trefethen.2008} 
for 
convergence statements.

\subsubsection{Clenshaw-Curtis quadrature}
\label{sec:Clenshaw-Curtis}

A numerical integration of (\ref{eq:complete_LK}) by an interpolatory quadrature rule replaces integrals by weighted sums from 
 values at certain grid points. If these grid points 
are the (Gauss-Lobatto) Chebyshev nodes $\{\tilde \theta_k\}_{k\in\{0,\ldots,N\}}$ introduced in (\ref{eq:ChebNodes}), this amounts to a Clenshaw-Curtis quadrature, cf.\ \cite[Sec.~3.7]{Funaro.1992}. 
The 
weights\footnote{The weights $w_k=\frac \delay 2 \int_{-1}^1 \ell_k(\tilde \vartheta)\,\mathrm d\tilde \vartheta$ 
are integrals of the Lagrange polynomials  in 
 $\int_{-\delay}^0 u(\theta)\,\mathrm d \theta\approx  \int_{-\delay}^0 \sum_{k=0}^{N} u(\tilde \theta_k )\ell_k(\vartheta(\theta))\,\mathrm d \theta =\sum_{k=0}^{N} u(\tilde \theta_k ) w_k$.
} 
\!\! $w_k$ are, e.g.,  available\footnote{implemented via \texttt{[theta,w]=chebpts(N+1,[-delay,0])}} 
in the Chebfun toolbox \cite{Driscoll.2014}.
For  (\ref{eq:complete_LK}), we obtain 
\begin{align} \label{eq:num_int_complete_LK}
V(\phi) 
&\approx
\phi^\top(0) P_{\mathrm{xx}}\,\phi(0)
+
2 \sum_{k=0}^{N} w_k  \hspace{0.5pt} \phi^\top\!(0) P_{\mathrm{xz}}(\tilde \theta_k) \hspace{0.5pt}
\phi(\tilde \theta_k)
\nonumber \\
&\quad +
  \sum_{j=0}^{N}  w_j \sum_{k=0}^{N} w_k  \hspace{0.5pt}
	\phi^\top\!(\tilde \theta_j)
	P_{\mathrm{zz}}(\tilde\theta_j,\tilde\theta_k) \hspace{0.5pt}
	\phi(\tilde\theta_k)
\nonumber  \\
&\quad + 
\sum_{k=0}^{N} w_k  \hspace{0.5pt}
\phi^\top\!(\tilde\theta_k)
P_{\mathrm{zz,diag}}(\tilde\theta_k) \hspace{0.5pt}
\phi(\tilde\theta_k).
\end{align}
Let $y^k=\phi(\tilde \theta_k)$, $k\in\{0\ldots,N\}$,   where 
$y^N=\phi(\tilde \theta_N)=\phi(0)$.
As in (\ref{eq:V_y_matrix}), the result (\ref{eq:num_int_complete_LK}) can be written as a quadratic form ({with $p=\dim(z):=\dim([{y^0}^\top,\ldots, {y^{N-1}}^\top]^\top)=nN$)
\begin{align}
&V(\phi)
\;\approx\;
y^\top P_{\ind{y}}^{\,\mathrm{quad} }\,y
\;=\;
y^\top
\,
\Bigg( \qquad 
\begin{bmatrix}
0_{p\times p} & 0_{p\times n}\\
 0_{n\times p} & P_{\mathrm{xx}}^{}
\end{bmatrix}
\nonumber\\
&
+
\begin{bmatrix}
0_{p\times (p+n)} 
\\[0.3em]
 {\rule[.5ex]{2.5ex}{0.3pt}}  
\; \;
 (P_{\mathrm{xz}^{}}(\tilde\theta_k)w_k)_k^{} 
\,\;
{\rule[.5ex]{2.5ex}{0.3pt}}
\end{bmatrix}
+
\begin{bmatrix}
\;\; \;0_{(p+n)\times p} 
&
\begin{matrix}
\rule[1.3ex]{0.3pt}{1.0ex} 
\\[-0.5em] 
 (w_j P_{\mathrm{xz}}^\top(\tilde\theta_j))_j^{} 
\\[-0.5em]
\rule[-1.5ex]{0.3pt}{1.0ex} 
\end{matrix}
\end{bmatrix}
\nonumber\\
&
+
\Big(w_j P_{\mathrm{zz}}(\tilde\theta_j,\tilde\theta_k) w_k\Big)_{jk}
\!+
\mathrm{blkdiag}\Big((w_k P_{\mathrm{zz,diag}}(\tilde \theta_k))_k^{}\Big)
\;
\Bigg)\; y.
\label{eq:P_num_int_complete_LK}
\\[-1.85em] \nonumber 
\end{align}
See 
(\ref{eq:P_quad_Clenshaw_Curtis_structure}) for a factorization taking (\ref{eq:P_in_completeTypeLK}) into account. 
Note that the right lower component of $P_{\ind{y}}^{\,\mathrm{quad} }$  approximately becomes $P_{\mathrm{xx}}$ since the other contributions are weighted by $w_N$, which  is quite small in the non-equidistant grid.

\subsubsection{Gauss quadrature}
\label{sec:Gauss}

As an alternative, 
we apply (Legendre) Gauss quadrature. Thus, the integral of a function is approximated by weighted sums from the function values at (Gauss) Legendre nodes. 
Being Gauss nodes, cf.\ Table \ref{tab:polynomialApprox},  they do not contain the boundary points of the domain $[-\delay,0]$.  
That is why we take $N$ (Gauss) Legendre nodes\footnote{via \texttt{[thetaL,wL]=legpts(N,[-delay,0])} using the toolbox \cite{Driscoll.2014}}
 $\tilde \theta^L_k$,   
and 
add the zero end point 
with zero weight to get the $N+1$ nodes 
$[(\tilde \theta^L)^\top,0]^\top$ and 
 weights 
$[(w^L)^\top,0]^\top$. Therefore, contrary to the Gauss-Lobatto-node-based Clenshaw-Curtis quadrature,  the contributions in (\ref{eq:P_num_int_complete_LK}) do not overlap, yielding for 
$y_G^{}=[\phi^\top(\tilde \theta^L_0),\ldots, \phi^\top(\tilde \theta^L_{N-1}), \phi^\top(0)]^\top$
\begin{align}
V(\phi)
\approx\;&
y_G^\top
\!\!
\begin{bmatrix}
\Big(w_j^L P_{\mathrm{zz}}(\tilde\theta_j^L,\tilde\theta_k^L) w_k^L\Big)_{jk}\!\!\!
+
\!D 
&
\begin{matrix}
\rule[1.3ex]{0.3pt}{1.0ex} 
\\[-0.5em] 
 (w_j^L P_{\mathrm{xz}}^\top(\tilde\theta_j^L))_j^{} 
\\[-0.5em]
\rule[-1.5ex]{0.3pt}{1.0ex} 
\end{matrix}
\\[1.5em]
 {\rule[.5ex]{2.5ex}{0.3pt}}  
\; \;
 (P_{\mathrm{xz}^{}}(\tilde\theta_k^L)w_k^L)_k^{} 
\,\;
{\rule[.5ex]{2.5ex}{0.3pt}}
&
P_{\mathrm{xx}}^{}
\end{bmatrix}
\! y_G^{}
\nonumber
\\
&\text{with}\; D=\mathrm{blkdiag}\Big((w_k^L P_{\mathrm{zz,diag}}(\tilde \theta_k^L))_k^{}\Big).
\label{eq:P_num_int_G_complete_LK}
\end{align}